\documentclass[runningheads,a4paper]{llncs}

\def\technicalreport{} 
\ifx\undefined\technicalreport
  \newcommand\tr[1]{}
  \newcommand\ntr[1]{#1}
\else
  \newcommand\tr[1]{#1}
  \newcommand\ntr[1]{}
\fi

\usepackage{amssymb}
\usepackage{graphicx}

\usepackage{amsmath}
\usepackage{amssymb}
\usepackage{graphicx}
\usepackage{theorem}
\usepackage{float}
\usepackage{algorithm}
\usepackage[noend]{algpseudocode} 
\usepackage{subfigure}
\usepackage{color}
\usepackage{colortbl}
\usepackage{xspace}

\renewenvironment{proof}{\noindent {\bf Proof: }}{\hfill\proofend\linebreak}
\newenvironment{proofof}[1]{\noindent {\bf Proof of #1: }}{\hfill\proofend\linebreak}
\newcommand{\proofend}{$\Box$}
\newcommand{\OO}{\ensuremath{{\cal O}}}

\newcommand\abs[1]{\left|#1\right|}

\newcommand{\vect}[1]{\mathbf{#1}}

\newcommand\dist[1]{\left\|#1\right\|}

\newcommand{\RNum}[1]{\uppercase\expandafter{\romannumeral #1\relax}}

\usepackage{color}

\usepackage{url}
\urldef{\mailsa}\path|{alfred.hofmann, ursula.barth, ingrid.haas, frank.holzwarth,|
\urldef{\mailsb}\path|anna.kramer, leonie.kunz, christine.reiss, nicole.sator,|
\urldef{\mailsc}\path|{janson, schindel}@informatik.uni-freiburg.de|    
\newcommand{\keywords}[1]{\par\addvspace\baselineskip
\noindent\keywordname\enspace\ignorespaces#1}

\begin{document}

\mainmatter  

\ntr{\newcommand\titleplain{Self-Synchronized Cooperative Beamforming in Ad-Hoc Networks}}
\tr{\newcommand\titleplain{Ad-Hoc Network Unicast in Time $\OO(\log\log n)$    using Beamforming}}

\title{\titleplain}

\titlerunning{\titleplain}

%
%
\author{Thomas Janson \and Christian Schindelhauer}
\authorrunning{Thomas Janson \and Christian Schindelhauer}

\institute{University of Freiburg,\\
Germany\\ 
\mailsc\\
}

\toctitle{\titleplain}
\tocauthor{Thomas Janson, Christian Schindelhauer}
\maketitle

\begin{abstract}
We investigate the unicast problem for ad-hoc networks in the plane using MIMO techniques. 
In particular, we use the multi-node beamforming gain and present a self-synchronizing algorithm for the necessary carrier phase synchronization.
First, we consider $n$ nodes in a grid where the transmission power per node is restricted to reach the neighboring node.
We extend the idea of multi-hop routing and relay the message by multiple nodes attaining joint beamforming gain with higher reception range.
In each round, the message is repeated by relay nodes at dedicated positions after a fixed waiting period.
Such simple algorithms can send a message from any node to any other node in time $\mathcal{O}(\log \log n - \log \lambda)$ 
and with asymptotical energy $\mathcal{O}(\sqrt{n})$, the same energy an optimal multi-hop routing strategy needs using short hops between source and target.
Here, $\lambda$ denotes the wavelength of the carrier. For $\lambda \in \Theta(1)$ we prove a tight lower time bound of $\Omega(\log \log n)$.

Then, we consider $n$ randomly distributed nodes in a square of area $n$ and we show for a transmission range of $\Theta(\sqrt{\log n})$ and for a wavelength of $\lambda = \Omega(\log^{-1/2}n)$ that the unicast problem can be solved in $\mathcal{O}(\log \log n)$ rounds as well. The corresponding transmission energy  increases to $\mathcal{O}(\sqrt{n} \log n)$.
%
 Finally, we present  simulation results   visualizing the nature of our algorithms.

\end{abstract}


\section{Introduction} \label{s:introduction}

Mobile devices reduce their wireless transmission power to prolong battery lifetime. An energy preserving extension of the transmission range is cooperative beamforming. Here, nodes  cooperate by sending the same message and produce together a stronger signal than a single node. Without further adaption the different positions of the senders result in a delay skew such that the signals may not be correlated at some receiver positions.
 When the sending times are coordinated we achieve the so-called beamforming, where the radiant sender beams result in a strongly correlated signal towards a certain direction. 
In \cite{JS13_Beamforming_Line}, we study fundamental features of phase-synchronized ad-hoc network nodes and show an exponential speedup for the broadcast operation of nodes placed  on a line.
Here, we are concerned in extending these observations to the two-dimensional plane. 

Unicast is defined as transfer of a message from a source node to a target node. For wireless communication the straight-forward solution is a direct transmission by increasing the signal strength at the sender such that the target node can receive the signal. While the message delay is optimal, the necessary transmission power is the drawback, since it quadratically increases with respect to the distance between sender and receiver. 

In a power constraint scenario 
direct communication is not always available. Then, routes with multiple hops  must be used. Messages are passed from the source via relay nodes towards the target. Regarding the sum of transmission energy, strategies with many short hops are better than single hop strategies.  On the other hand, the delay increases with the number of hops. Here, we consider networks with $n$ nodes in the plane placed on a  $\sqrt{n} \times\sqrt{n} $ quadratic grid with unit distance between neighbored nodes. The delay or routing time for multi-hop routing with distances $1$ each is $\OO\left(\sqrt{n}\right)$. The energy consumption compared to direct communication decreases by a factor of  $\OO\left(1/\sqrt{n}\right)$.

Multi-hop routing implements time multiplexing, i.e. using several time slots, and spatial multiplexing by blocking a smaller area for communication compared to direct communication. 
However, the simultaneously sending nodes can do much better when one uses cooperative beamforming.
One might expect that doubling the power of two senders increases the transmission range by a factor of $\sqrt2$. 
However, the superposition principle for electric fields implies that the signal strengths add up and this strength is proportional to the square root of the transmission energy. Therefore, the reception range of two close phase-synchronized senders increases by a factor of two  \cite{JS13_Beamforming_Line}.

This is the beamforming aspect of MIMO (multiple input/multiple output) technology in the line of sight case. 
Besides beamforming, MIMO allows to establish parallel channels with $n$ senders (input) and $m$ receivers (output), resulting up to $\min\{n,m\}$ parallel transmission channels. 
For this it is necessary that sig\-nals are reflected from obstacles in the environment, if the sender and receiver antennas are distant. \tr{The channel matrix $H$ describes for each sender/receiver pair the attenuation and phase shift between them. If  this matrix $H$ shows many large eigenvalues, then parallel channels enable increased throughput, in addition with suitable encoding and decoding. 
}\ntr{However}\tr{Therefore}, MIMO signal processing is complex and MIMO does not work in the line-of-sight scenario with distant sender and receiver antenna arrays unlike beamforming. 

In this paper we consider the line-of-sight model\ntr{ and beamforming. }\tr{, where the channel matrix $H$ does not allow multiple channels. Therefore, beamforming is the focus of this paper. }It is achieved by adjusting the sender time points such that the received signal consists of synchronized signals which add up because of the superposition principle. A message can be received if this signal strength is larger than a given value, i.e. the signal-to-noise ratio threshold.
%


Our main method is to assign rectangular areas for suitable relay nodes. These nodes cooperate for the beamforming of the unicast message. For this, nodes store the received message and resend it at time points depending on the reception times. We restrict the corresponding transmission power such that each node can only reach its neighborhood  without beamforming. The overall goal is to minimize the transmission time of a single unicast message.

\ntr{
Due to page limitations some proofs are presented in a technical report \cite{JS14_Beamforming_LogLog_TR}.}
%


\section{Related Work}

Gupta and Kumar  \cite{Gupta00thecapacity} analyze  the throughput capacity of wireless networks. The throughput capacity of a network node specifies the average data rate to a communication partner multiplied by the communication distance. 
For the case of nodes positioned independently at random in the plane and random communication pairings, they show that the capacity is  $\Theta(\frac{1}{\sqrt{n \log n}})$ in the best case. Here, multiple hop routes using next neighbors turn out to be the best choice. 
It turns out that the communication bottleneck is a cut through the middle of the network, on which each node has to uphold $\mathcal{O}\left(\sqrt{n}\right)$ connections throttling the throughput by a factor of $\mathcal{O}(\frac{1}{\sqrt{n}})$.
%
%
It is necessary to increase the sending power by $\mathcal{O}\left(\log n\right)$ to guarantee network connectivity with high probability.
By this, the throughput is further reduced by a factor of $\mathcal{O}(\frac{1}{\sqrt{\log n}})$.
In such a model,  our beamforming approach reaches only a throughput capacity comparable to direct point-to-point communication. Yet, for a scenario with only one point-to-point communication, where the transmission power is limited to $\Theta(\frac{\log n}{n})$ (the best case of \cite{Gupta00thecapacity}), the multi-hop scheme 
 has a throughput of $\Theta(\frac{\sqrt{\log n}}{\sqrt{n}})$, while our unicast has a throughput of $\Theta(\frac{1}{\log\log n})$.

In  \cite{JS13_Beamforming_Line} we present broadcasting algorithms for nodes on a line in the line-of-sight case.  
We prove that broadcasting can be done in   $\OO\left(\log n\right)$ rounds for $n$ nodes regularly placed on a line, where each node alone can only reach its next neighbor. This is obtained by the beamforming gain and on-the-fly synchronization using only the reception time of the message. 
This scheme produces only constant factor increase of the energy consumption compared to direct neighbor communication, which needs $\OO\left(n\right)$ rounds. Here, we consider the two-dimensional setting for the same model
and reuse the one-dimensional variant as a startup sub-routine. 


\tr{In \cite{2012RandomMIMOJanson} we analyze beamforming gain for antennas placed in an area. We estimate the angle of the main beam for $m$ randomly placed senders in a disk, which has size $\OO(\lambda/d)$ where $d$ is the diameter of the disk and $\lambda$ the carrier wavelength.
We find side beams within an angle of $\OO(\lambda/(d \sqrt{m}))$. Towards other directions, the signal strength is reduced to an expected size of $\sqrt{m}$ times the sender's signal strength, while the main beam is $m$ times larger than each sender's signal strength.}


In \cite{NGS09_Linear_Capacity_Beamforming,ozgur2010linearCapacity} communication schemes are presented that achieve order-optimal throughput by using MIMO techniques. Here, nodes in designated areas cooperate in order to increase the communication capacity resulting in higher bandwidth or increased transmission radius.  
In \cite{NGS09_Linear_Capacity_Beamforming} the beamforming gain is exploited at designated areas of relay nodes between sender and receiver. In \cite{ozgur2010linearCapacity} diversity gain of highly parallel MIMO channels is used.
An important step in many MIMO protocols is encoding and decoding the transmitted signal, which needs additional communication at the sender and receiver side. In practice, this is achieved by wiring the sender/receiver antennas into one device.  For ad-hoc networks this step has to be emulated via wireless communication. The authors use a hierarchical approach, where the communication for the encoding at the sender nodes is organized by a recursive algorithm (and vice versa for the decoding at the receiver nodes). If this step can be done without a substantiate increase of the original message size (which may be doubted), then this achieves a capacity and time gain.
%
%
The transmission time is $\OO\left(\log n\right)$, which corresponds to the number of hierarchical steps and the capacity is up to linear depending on the path loss model. However, a minimum message length is required depending on the capacity and the authors  assume a  channel matrix with large eigenvalues, in contrast to the free-space model underlying this work. Here, we solve unicasting in time $\OO(\log \log n)$ and the algorithms  presented here are much simpler, since they do not use any MIMO encoding/decoding.
%
 


The authors of \cite{mlo13_telescopic_beamforming} use a similar approach by using beamforming of rectangular areas. Their algorithm spreads the information to a telescope-like region with increasing adjacent rectangles. Then, a mirrored construction is appended in order to reach the target node. They conclude that the beamforming gain is maximized up to a constant factor at each receiver as long as the area size of beamforming nodes is much smaller than $\sqrt{n}$ for $n$ nodes in the network. The authors cannot give a closed form for the dimensions of the rectangles and refer to a Matlab program computing optimal sizes. An important difference to our approach is that they allow additional transmission power $a>1$ for a short period $1/a$. Interestingly, their choice is $a = \Theta(1/n^{2/3})$ which results in throughput $T= \OO( n^{2/3})$. We show that the choice of adjacent rectangles might be problematic, since our simulation results indicate that some receivers in the adjacent rectangle might not be reached. In this paper, we emphasize the large influence of the carrier wavelength and present a closed-form solution for the placement and dimensions of rectangular beam-forming areas. 
Furthermore, we present a solution which does not need the full channel state information.


\section{Physical Model}

The signal quality and the related transmission bandwidth of a communication channel between sender and receiver is difficult to model because of many effects arising in practice, e.g. multi-path propagation, diffraction, changing environment, node movement, etc. We neglect these effects and use the free-space model, where the signal strength as a function of the position of nodes in the network.  
Following \cite{Tse_fundamentals_book}, the signal output $y$ at the receiver depends on the signal inputs at senders $x_1, \ldots, x_m$ as 
\begin{eqnarray}
y = \sum_{i=1}^m h_i \cdot x_i \ . \label{eq:super}
\end{eqnarray}
This establishes the physical input-output-model of a MISO channel (Multiple Input Single Output).  Inputs and outputs are seen from the communication channel and not from the senders or receivers.
We assume that all nodes emit the same input signal $x=x_i$ with the same transmission power but with a time shift in order to correlate the phases resulting in a beamforming gain at the target with output $y$. We denote by $j$  the imaginary number  ($j^2 =-1$).
The baseband channel gain $h_i$ for the $i$-th sender node is
\begin{eqnarray}
h_i &=& \frac{1}{\dist{\vect{u_i},\vect{v}}} \cdot e^{\textstyle -\frac{j2\pi}{\lambda} \cdot \dist{\vect{u_i},\vect{v}}}.
\end{eqnarray}
The attenuation factor $\dist{\vect{u_i},\vect{v}}^{-1}$ describes the path loss depending on the distance  $\dist{\vect{u_i},\vect{v}}$ between the nodes at positions $u_i$ and $v$.  Since the power is proportional to the square of the signal strength this corresponds to the standard energy path loss model for line-of-sight and the far-field assumption with $\dist{\vect{u_i},\vect{v}} > 2\lambda$ where the energy decreases proportional to $\dist{\vect{u_i},\vect{v}}^{-2}$.
The wavelength $\lambda=c/f$ of the carrier frequency $f$ plays an important role for the beamforming. We denote by $c$ the speed of light. In~\cite{2012RandomMIMOJanson} we show that the sender geometry and the wavelength determine the width of the main beam, as well as the size of side beams. The distance between sender and receiver also results in a phase shift described by a rotation of the signal in complex space.

This signal value describes the electric field produced by the sender, and by the superposition principle the resulting field is the sum of the signals in Equation~(\ref{eq:super}).
Interfering radio signals and errors occurring during the modulation and demodulation  are modeled as being uncorrelated to the line-of-sight signal as additive white Gaussian noise $w$,  which is Gaussian distributed $w \sim \mathcal{N}\left(0,\sigma^2\right)$ with variance $\sigma^2$. So, the received signal is described by $y+w$.

 A signal can be received if the signal to noise ratio is larger than a threshold $\tau$, i.e.
$\text{SNR} = \frac{P}{N}\ \geq \ \tau \ , $ where $N$ is the energy of the noise.

We restrict the transmission power for each node in the grid such that only the vertical and horizontal neighbors in distance can be reached, if only a single sender is active. The received signal power is modeled by $P=|y|^2$.

So, we choose $\tau = 1$  and $|x_i| \leq 1$ to describe the situation in the grid. We also consider the random placement model, where we randomly position $n$ nodes into a grid of area $n$. In~\cite{gupta1998critical} it is shown that the minimum transmission distance for achieving connectivity in this model is $\Omega(\sqrt{\log n})$. Therefore, we increase the maximum size  $|x_i|\leq  k (\log n)^{1/2} $  of the signal and let $\tau =1$ for some constant $k$.
%

According to the Shannon-Hartley theorem, it is possible to achieve an information rate of  $ B\cdot\log\left(1+\text{SNR}\right)$. So, a higher signal-to-noise ratio can increase the information rate. This effect is not used in this work, since at the relevant receiver antennas the received signal power is close to the SNR threshold.

\section{Loglog n Unicast} \label{sec:loglog_unicast_algorithm}


The basic idea of our unicast algorithm is a multi-hop algorithm with relays between sender and receiver shown in Figure~\ref{fig:Unicast_multihop}, but with the special property that each relay consists of multiple nodes which cooperate to perform joint sender beamforming, see Figure~\ref{fig:Rectangle_Beamforming_Gain}.
\begin{figure}[htb]
  \centering
  \subfigure[Multi-hop between rectangles of beamforming senders.]{
    \label{fig:Unicast_multihop}
   \includegraphics[scale=0.55]{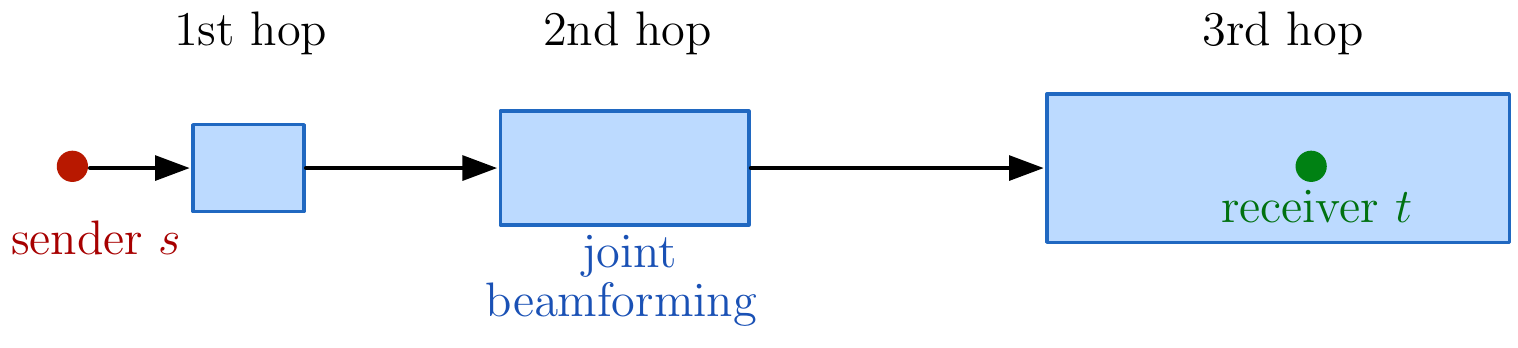}
  }\hspace{0.5em}
  \subfigure[Beamforming from sender to receiver rectangle]{
    \label{fig:Rectangle_Beamforming_Gain}
    \includegraphics[scale=0.55]{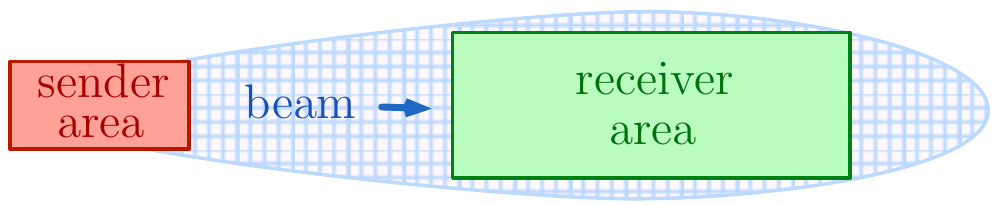}
  }
  \caption{Scheme of the $\OO(\log \log n)$-Unicast algorithm.}
  \label{fig:Unicast_problem_setting}
\end{figure}
With beamforming gain, the hop distance increases double exponentially such that this unicast algorithm needs $\OO\left(\log\log n\right)$ hops from the source to the target.

We use beamforming for sending\tr{\footnote{ 
We make no use of receiver beamforming (SIMO). It requires that cooperative nodes exchange the received signals as quantized data for signal processing. For large sets of receivers, multi-hop transmissions \tr{of the signals }are necessary, and more over, the message size grows exponentially if receiver beamforming is applied recursively\tr{ since quantizing the signal has to be applied recursively}.}} (MISO) which requires, when performed with several senders in parallel, the distribution of the  message to all senders and phase synchronization between all senders. As Figure~\ref{fig:Rectangle_Beamforming_Gain} indicates, we will show that we can broadcast a message from a sender to a receiver area with rectangular shape such that all nodes in the receiver area have the same message for cooperated sender beamforming in the next round.
 For synchronizing the sender phases, we present two algorithms. Algorithm~\ref{alg:unicast1} corrects the phase at the relay nodes using the position of the nodes, whereas Algorithm~\ref{alg:unicast2} is self-synchronizing.  Algorithm~\ref{alg:unicast1} outperforms Algorithm~\ref{alg:unicast2} regarding the transmission time by a constant factor.

%
%

We first describe the $\OO(\log \log n)$-unicast algorithm in a network with $\sqrt{n}\times\sqrt{n}$ nodes placed in a grid. For unit grid distance we assume $\lambda\le\frac{1}{2}$ to meet the far-field assumption.
We start to describe the algorithm for a message transmission along the $x$-axis in the middle of the grid and generalize it for other coordinates, later on. The source node is at $s=(0, 0)$ and the target node at $t=(\sqrt{n},0)$.
The algorithm consists of two phases, an initial phase (Fig.~\ref{fig:Unicast_multihop} 1st hop) where we broadcast the message from the source to the first rectangle of relay nodes, and a second phase where we perform multi-hop with distributed beamforming (Fig.~\ref{fig:Unicast_multihop}, 2nd, 3rd hop). The required rectangular area to be informed in phase 1 follows from the requirements of phase two, and thus we present phase 2 first.

We first describe how to set up phases for distributed beamforming when the senders are placed on a line along the  $x$-axis (see Fig.~\ref{fig:Line_Beamforming_Sync}) and extend that for  rectangles in the plane, later on.
\begin{figure}[hbt]
	\begin{center}
	\includegraphics[scale=0.58]{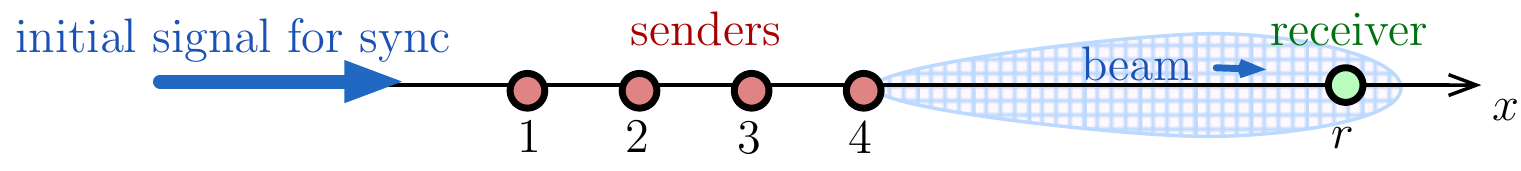}
	\end{center}\vspace{-1em}
	\caption{Synchronization in the one-dimensional case}
	\label{fig:Line_Beamforming_Sync}
\end{figure}
Assume we have senders placed at $\left(i,0\right)$ with $1\le i\le n$ performing beamforming to a receiver $r$ at $\left(r_x,0\right)$ with $r_x >n$.
To attain full beamforming gain, the senders start the transmission with a delay of $\left(n-i\right)/c$ for propagation speed $c$ such that all transmissions arrive exactly at the same time and consequently in the same phase. We synchronize all senders with the initial signal containing the message. It is received at a node placed at $\left(i,0\right)$ at time $t = i/c$ and  if each node resends the message immediately, it sends the message with delay $-i/c$, which is the desired beamforming setup to receiver $r$. Hence, broadcasting along a line achieves self-synchronization for distributed beamforming.

For beamforming senders in a rectangle, we use the same synchronization setup, and each node $u$ at coordinates $\left(u_x,u_y\right)$ sends at time $t=u_x/c - t_0$ which only depends on the $x$-coordinate and offset time $t_0$ has to be chosen such that the sender with smallest $u_x$ sends at time $t=0$ without delay. If it holds $\dist{u,r} = r_x-u_x$, which is the case for nodes along the $x$-axis, the synchronization is perfect. But for a rectangular area of nodes with width $w_i$ and height $h_i$, the reception delay depends also on the $y$-coordinate.
The delay function $\psi\left(i,r\right)$ computes for a receiver at coordinates $\vect{r}=\left(r_x,r_y\right)$ the delay to attain synchronization, which is phase angle $\arg[e^{-j2\pi r_x/\lambda}]$.
\begin{eqnarray}
\psi\left(i,\vect{r}\right) &=& \frac{1}{f} + \frac{1}{2\pi f} \arg\left[ \sum_{s\in\left(w_{i-1}\times h_{i-1}\right)} \frac{e^{-j2\pi \left(\dist{\vect{s},\vect{r}}-r_x\right)/\lambda}}{\dist{\vect{s},\vect{r}}} \right] \label{eq:rectangle_delay}
\end{eqnarray}
When applying delay $\psi\left(i,r\right)$ at each receiver $r$, all nodes are synchronized for beamforming such that each node $r$ sends with a delay of $-r_x/c$.
By a proper choice of the dimensions of the rectangles $\left(w_i,h_i\right)$, we can assure that the phase shift 
is less than $\pi/2$ and thus $\psi\left(i,\vect{r}\right)>0$ (compare Lemma~\ref{Le:constant_phase_shift}). 

This leads to Algorithm~\ref{alg:unicast1} where the delay $\psi\left(i,r\right)$ is used in line~\ref{alg_unicast1_delay} in order to synchronize the receivers in the $i$-th round for the  $w_i \times h_i$-receiver area. The if-condition in Line~\ref{alg_unicast1_receiverFilter} assures that only receivers in the correct receiver area process the message.
\ntr{\vspace{-.5em}}
\begin{algorithm}[h]
\caption{Unicast \RNum{1} }\label{alg:unicast1}
\begin{algorithmic}[1]
\Procedure{receive}{receiver $r$, message $m$, time $t$}
   \If{ \Call{isInRectangle}{$\text{round}\left(t\right)$, r}} \Comment{only process in active rectangle} \label{alg_unicast1_receiverFilter}
   	\State \Call{wait}{$\psi\left(\text{round}\left(t\right),r\right)$}	\Comment{phase correction} \label{alg_unicast1_delay}
   	\State \Call{send}{m}		\Comment{coordinated beamforming sending}
   \EndIf
\EndProcedure \vspace{0.2em}
\Function{isInRectangle}{round $i$, position $p$} \Comment{true for active receivers}
\State \Return $w_0+w_i + 2\sum_{k=1}^{i-1}w_k \le p_x \le w_0+2\sum_{k=1}^{i}w_k \text{ \& }  0 \le p_y \le h_i $
\EndFunction
\end{algorithmic}
\end{algorithm}
\ntr{\vspace{-.5em}}

The following Lemmas~\ref{Le:constant_phase_shift}-\ref{le:rectangle_dim} specify the dimensions and distances between rectangles of relay nodes where the multi-hop procedure of Algorithm~\ref{alg:unicast1} with distributed sender beamforming is possible.

\begin{lemma} \label{Le:constant_phase_shift}
If a single sender $s$ sends a signal to a $w \times h$ rectangular area in a distance of at least $w$ (see Figure~\ref{fig:Unicast_constant_delay}), then the phase shift with respect to the phase $2 \pi r_x/\lambda$ is at any receiver node $r$ inside the area at most $\alpha$ if $h^2 \le \frac{\alpha}{\pi} \lambda w$.
\end{lemma}
%
\begin{figure}[h]
	\begin{center}
	\includegraphics[scale=0.37]{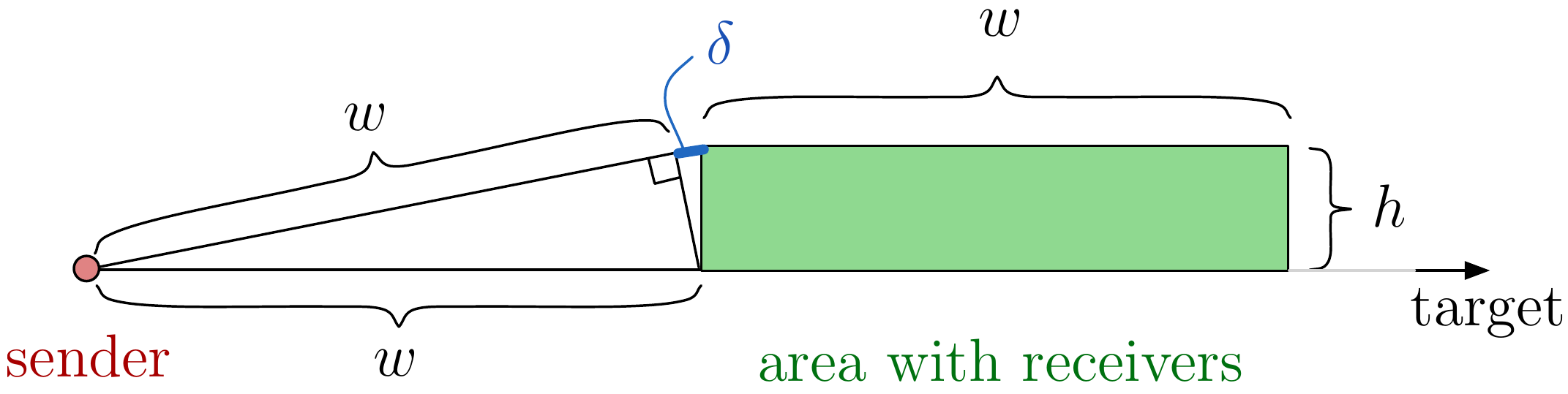} 
	\end{center}
	\caption{Broadcast of an single sender (red) to receivers in the green area.}
	\label{fig:Unicast_constant_delay}
\end{figure}
\begin{proof}
Let $x$ denote the signal of the sender $s$ and $y$ the signal at $r$. Then,
$$ y = \frac{x}{\dist{\vect{s},\vect{r}}} \cdot e^{\textstyle -\frac{j2\pi}{\lambda} \cdot \dist{\vect{s},\vect{r}}} \ .$$\vspace{-0.5em} 

Thus, the phase shift is described by 
$-\arg(\frac{y}{x}) = \frac{2\pi}{\lambda} \dist{\vect{s},\vect{r}}$. The difference of phase shifts is therefore 
\begin{eqnarray*}
 \delta &=& \frac{2\pi}{\lambda} \dist{\vect{s},\vect{r}}-\frac{2 \pi x}{\lambda} \ 
= \frac{2 \pi}{\lambda} \left(\sqrt{r_x^2+r_y^2}- r_x\right) 
 = \frac{2 \pi}{\lambda} r_x \left(\sqrt{1+\left(\frac{r_y}{r_x}\right)^2}-1\right) \ .
\end{eqnarray*}
\ntr{
This phase difference is maximized for $r_y = h$ and $r_x=w$ and by applying the relation $\sqrt{1+x^2}-1 \le \frac{x^2}{2}$ for all $x\geq0$ (see~\cite{JS14_Beamforming_LogLog_TR}) we get
\begin{eqnarray*}
\delta \leq \frac{\pi}{\lambda} \frac{r_y^2}{r_x} = \frac{\pi}{\lambda} \frac{h^2}{w} \ .
\end{eqnarray*}
}\tr{
We can apply Lemma~\ref{Marx-a} given in the Appendix and get 
$$  \delta \leq \frac{\pi}{\lambda} \frac{r_y^2}{r_x} \ .$$
This phase difference is maximized for $r_y = h$ and $r_x=w$. Then, 
$$ \delta \leq \frac{\pi}{\lambda} \frac{h^2}{w} \ .$$} 
From $h^2 \le \frac{\alpha}{\pi} \lambda w$ it follows that $\delta \leq \alpha$.
\end{proof}
Note that the difference between the signal and the offset is so small, e.g. for $\alpha \leq \pi/4$, that it is less than one wavelength. So, if we repeat the message transmission after a fixed time offset in the next round, then the message modulated upon the carrier wave is in sync with all the other sender nodes provided by using the same time offset.

\begin{lemma} \label{Le:Area_growth_known_position}\label{l:twos}
A $w_i\times h_i$-rectangular area of beamforming senders $S$ can reach any node in a $w_{i+1} \times h_{i+1}$ rectangle at distance $w_{i+1}$ if
\begin{eqnarray}
h_{i+1} &\geq& h_i \ , \label{l:eq2}\\
w_{i+1} &\geq& w_i \ , \label{l:eq3}\\
w_{i+1} &\leq& \frac1{3\sqrt2} w_i h_i \ , \label{l:eq4}\\
h_{i+1} &\leq& w_{i+1} \ , \ \text{and} \label{l:eqRatio}\\
h_{i+1}^2 &\leq& \frac1{4} \lambda w_{i+1}\ . \label{l:eq5}
\end{eqnarray}
\end{lemma}
\begin{proof}
%
Remember that all sending nodes of a vertical column in the grid  have the same phase. 
The received signal $y$ at node $r$ is 
\tr{
\begin{eqnarray*} y &=&
 \sum_{s\in S} x_s \frac{1}{\dist{\vect{s},\vect{r}}} \cdot e^{\displaystyle -\frac{j2\pi}{\lambda} \cdot \dist{\vect{s},\vect{r}}}\\
 &=&
 \sum_{s\in S} e^{j  \frac{2 \pi u_x}{\lambda}} \frac{1}{\dist{\vect{s},\vect{r}}} \cdot e^{\displaystyle -\frac{j2\pi}{\lambda} \cdot \dist{\vect{s},\vect{r}}} \\
 &=&
 \sum_{s\in S} \frac{1}{\dist{\vect{s},\vect{r}}} \cdot e^{\displaystyle -\frac{j2\pi}{\lambda} \cdot \dist{\vect{s},\vect{r}}+j  \frac{2 \pi u_x}{\lambda}} \ .
 \end{eqnarray*}
 } \ntr{
 \begin{eqnarray*} y &=&
 \sum_{s\in S} x_s \frac{e^{ -\frac{j2\pi}{\lambda} \cdot \dist{\vect{s},\vect{r}}}}{\dist{\vect{s},\vect{r}}}  
 = \sum_{s\in S} e^{j  \frac{2 \pi u_x}{\lambda}} \frac{e^{ -\frac{j2\pi}{\lambda} \cdot \dist{\vect{s},\vect{r}}}}{\dist{\vect{s},\vect{r}}}  
 = \sum_{s\in S} \frac{e^{ -\frac{j2\pi}{\lambda} \cdot \dist{\vect{s},\vect{r}}+j  \frac{2 \pi u_x}{\lambda}} }{\dist{\vect{s},\vect{r}}} \ .
 \end{eqnarray*} }
 And from 
 Lemma~\ref{Le:constant_phase_shift} we get ($\alpha=\pi/4$) for
 $$ \beta_{s,r} := \frac{2 \pi u_x}{\lambda} -\frac{2\pi}{\lambda} \cdot \dist{\vect{s},\vect{r}}$$ from 
 $w_i \leq w_{i+1}$ and inequality (\ref{l:eq5})
 \begin{equation}
 0 \leq \beta_{s,r} \leq \frac{\pi}{4}\ . \label{eq:phase_shift_range}
 \end{equation}
 We want to prove that $|y|^2 =\text{SNR} \geq \tau =1$. For this it suffices to prove that for the real part of $y$, i.e. that
 $\Re(y) \geq 1$, since $|y|^2 = \Im(y)^2 + \Re(y)^2$. 
 
Using, $\dist{\vect{s},\vect{r}} \leq w_i + 2w_{i+1}\leq 3 w_{i+1} \stackrel{\text{by} (\ref{l:eq4})}{\leq} \frac{1}{\sqrt2} w_i h_i = \frac{1}{\sqrt2}|S|$ we get
\begin{eqnarray*} \Re(y) 
= \sum_{s\in S} \frac{\Re(e^{-j \beta_{s,r}}) }{\dist{\vect{s},\vect{r}}} 
= \sum_{s\in S} \frac{\cos \beta_{s,r} }{\dist{\vect{s},\vect{r}}} 
\geq \sum_{s\in S} \frac{1}{w_i+2 w_{i+1}} \cos \frac{\pi}{4} 
\geq \frac{w_i h_i}{3 w_{i+1}} \frac1{\sqrt2}  
 \geq  1.
 \end{eqnarray*}
\end{proof}
Figure~\ref{fig:Broadcast_step_growth} illustrates the relation between the sender and the receiver area. The delay $\delta$ illustrates the largest possible value $\beta_{s,r}$ in the range of Eq.~(\ref{eq:phase_shift_range}).
\begin{figure}[h]
	\begin{center}
	\includegraphics[scale=0.33]{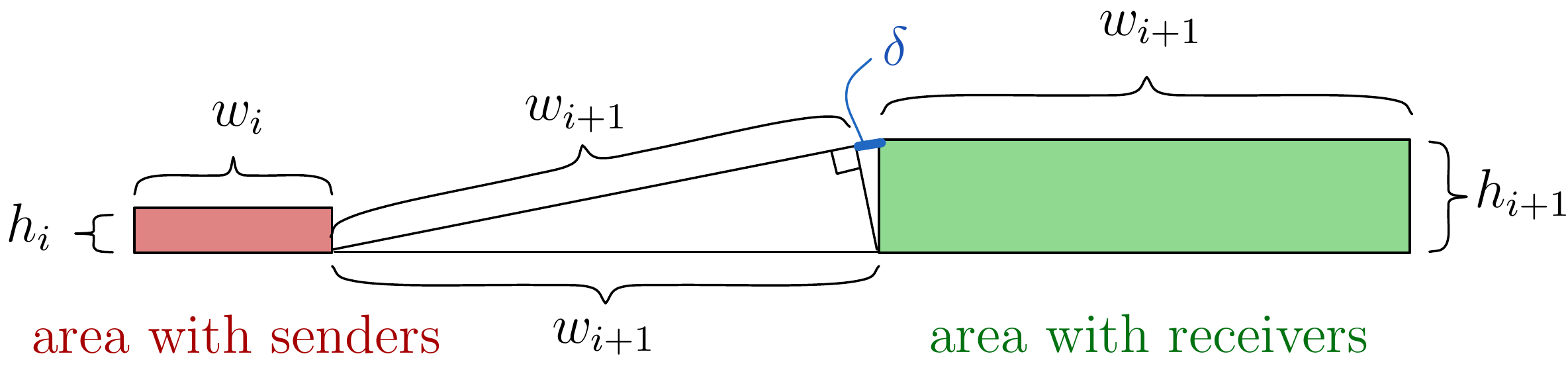}
	\end{center}
	\caption{Area growth during broadcast step.}
	\label{fig:Broadcast_step_growth}
\end{figure}
If the sender and the receiver are at the margin of the grid, we cannot expand the height of the relay node areas symmetrically along the line of sight between sender and receiver.
To apply the algorithm also at the margin of the network, we only expand the height of the rectangle in one direction, i.e. towards the center of the network. This has been already addressed in Equation~(\ref{l:eq5}).

%

This leads to the double exponential growth of the rectangles given in closed form in the following lemma.
\begin{lemma} \label{le:rectangle_dim}
The equations 
\begin{eqnarray}
w_{i} &=& \left(\frac{72}{\lambda}\right)
 \left(\frac{\lambda}{72}w_0 \right)^{(3/2)^{i}} \ ,\label{Eq:width}\\
h_{i} &=& \sqrt{18} \left(18^{-\frac{1}{2}} h_0\right)^{(3/2)^{i}} \ ,
\end{eqnarray}
for $i \in \{1,2, \ldots \}$
satisfy inequalities  (\ref{l:eq2}-\ref{l:eq5}) 
for $h_0 \geq 18^{\frac12}$, $w_0 \geq \frac{72}{\lambda}$ 
and
$h_0^2 = \frac14 \lambda w_0$.
\end{lemma}
\ntr{The proof can be found in \cite{JS14_Beamforming_LogLog_TR}.}
\tr{\begin{proof}
\begin{enumerate}
\item[(\ref{l:eq2})]: Merging Inequality~(\ref{l:eq4}) with~(\ref{l:eq5}) gives $h_{i+1} \le \frac{\lambda}{4}\cdot\frac{1}{\sqrt{18}}\cdot w_i h_i \le \frac{1}{\sqrt{18}} h_i^{3/2}$.
Then, $h_i \leq h_{i+1}$ follows from $h_0   \geq \sqrt{18}$.
\item[(\ref{l:eq3})]: $w_i \leq w_{i+1}$
is true if 
$w_0 \geq \frac{72}{w_0}$.
\item[(\ref{l:eq4})]:  $w_{i+1} \leq \frac1{3\sqrt2} w_i h_i $

Now $h_0 = 2 \sqrt{\lambda w_0}$, which implies
$ w_0 h_0 = 2 \sqrt{\lambda} w_0^{3/2}$.
Therefore,
\begin{eqnarray*}
\frac1{\sqrt{18}} w_i h_i
& = & \frac{72}{\lambda} \left(\frac{\lambda}{72}
\sqrt{\frac{\lambda }{72}} w_0^{3/2}\right)^{(3/2)^{i}}\\
& = & \frac{72}{\lambda} \left(\frac{\lambda  }{72}w_0\right)^{(3/2)^{i+1}}\\
& = & w_{i+1}
\end{eqnarray*}
\item[(\ref{l:eq4})] $ h_{i}^2 \leq \frac1{4} \lambda w_{i}$:
The following equations finalize the proof.
%
%
\begin{eqnarray*}
h_0^2 &=& \frac14 \lambda w_0 \\
\frac1{18} h_0^2 &=& \frac{\lambda w_0}{72} \\
  18 \left(\frac1{18} h_0^2\right)^{(3/2)^{i}}  &=& \frac14 \lambda\left(\frac{72}{\lambda}\right)
 \left(\frac{{\lambda}  }{{72}}w_0 \right)^{(3/2)^{i}} \\
\left.h_i\right.^2 & = & \frac14 \lambda w_i
\end{eqnarray*}
\end{enumerate}
$ $
%
%
%
\end{proof}
} 

So far, we assume that after the receipt of a message the relay node calculates the received phase from the senders' positions and readjusts the phase such that all vertical nodes are in phase. This step is not necessary, if the dimensions of the rectangles are chosen according to Lemma~\ref{Le:Area_growth_unknown_position}. Then, the received signal can be sent without phase correction from each relay node. The algorithm then reduces to two steps: If a message has been received, relay nodes check from the message header whether they are in the correct rectangles. Then, each relay node repeats the messages after the same time offset.

\ntr{\vspace{-0.5em}}
\begin{algorithm}[h]
\caption{Unicast \RNum{2} }\label{alg:unicast2}
\begin{algorithmic}[1]
\Procedure{receive}{receiver $r$, message $m$, time $t$}
   \If{\Call{isInRectangle}{$\text{round}\left(t\right)$, r}} \Comment{only process in active rectangle}
   	\State \Call{send}{m}		\Comment{coordinated beamforming sending}
   \EndIf
\EndProcedure
\end{algorithmic}
\end{algorithm}
\ntr{\vspace{-0.5em}}

\ntr{
\vspace{-1em}
\begin{lemma} \label{Le:Area_growth_unknown_position}
If the phase errors are not corrected in this routing, then the correct signal can be received if we use the following inequality instead of (\ref{l:eq5}).
\begin{eqnarray}
h_{i}^2 &\leq& \frac{3}{2\pi^2} \frac{1}{(i+1)^2}  \lambda w_{i}\ . \label{l:eq9}
\end{eqnarray}
\end{lemma}
} 
\tr{
\begin{lemma} \label{Le:Area_growth_unknown_position}
If the phase errors are not corrected in this routing, then the correct signal can be received if the following inequalities for the dimensions $h_i$ and $w_i$ of the relay rectangles are satisfied.
\begin{eqnarray}
h_{i+1} &\geq& h_i \ , \label{l:eq6}\\
w_{i+1} &\geq& w_i \ , \label{l:eq7}\\
w_{i+1} &\leq& \frac{1}{3\sqrt2} w_i h_i\label{l:eq8} \ , \\
h_i &\le&w_i \label{eq_ratio2} \ , \ \text{and} \\
h_{i}^2 &\leq& \frac{3}{2\pi^2} \frac{1}{(i+1)^2}  \lambda w_{i}\ . \label{l:eq9}
\end{eqnarray}
\end{lemma}
} 
The main idea is that the phase shifts in each round form a convergent series $\alpha_i = \frac{3\pi}{2} \cdot\frac{1}{\pi^2 i^2}$, such that the sum of all phases $\sum_{i=1}^{r} \alpha_i \leq \frac{\pi}{4}$ can be bound. \tr{The proof is otherwise analogous to Lemma~\ref{l:twos} and is combined with the proof of Lemma~\ref{Le:Rectangle_recursion_selfsync}.}

The dimensions of these rectangles can be chosen as follows.
\begin{lemma} \label{Le:Rectangle_recursion_selfsync}
The following recursions satisfy equations (\ntr{\ref{l:eq2}-\ref{l:eqRatio},}\tr{\ref{l:eq6}-}\ref{l:eq9}) for $h_0^2= \frac{3}{2\pi^2} \lambda w_0$ for $w_0  \ge\frac{96\pi^2e\cdot c_4}{\lambda}$, and $h_0 \geq 4\sqrt{18}$.
\begin{eqnarray}
w_{i+1}&=&\frac1{\sqrt{12}\pi}\cdot \frac{\sqrt{\lambda}}{i+1} \cdot w_i^{3/2} \label{eq:recursion_width_PhaseConversion} \\
h_{i+1}&=& 18^{-\frac14} \frac{1+i}{2+i} \cdot h_i^{3/2} \ . \label{eq:recursion_height_PhaseConversion}
\end{eqnarray}
The recursions are satisfied by the following equations.
\begin{eqnarray}
w_i &\leq& \left(\frac{\sqrt{\lambda}}{\sqrt{12}\pi}\right)^{2\left(3/2\right)^i-2} \cdot c_2^{-\left(3/2\right)^i} \cdot w_0^{\left(3/2\right)^i} \text{ with } c_2\ge 12.011 \label{eq:width_sync_closedform} \\
w_i &\geq& \left(\frac{\sqrt{\lambda}}{\sqrt{12}\pi}\right)^{2\left(3/2\right)^i-2} \cdot c_3^{-\left(3/2\right)^i} \cdot w_0^{\left(3/2\right)^i} \text{ with } c_3\le  1.58 \label{eq:width_sync_closedform} \\
h_i &=& 18^{\frac{-\left(3/2\right)^i + 1}{2}} \cdot \left(\frac{i+1}{i+2}\right)^{\frac{1}{2}\left(i-1\right)\cdot i} \cdot h_0^{\left(3/2\right)^i} \label{eq:height_sync_closedform}
\end{eqnarray}
\end{lemma}
Remember that we reach the constant length $w_0$ in a logarithmic number of rounds and therefore $\log_{3/2}\left(w_0\cdot\lambda\right) = 25 $ for a moderate expansion\tr{ with basis $3/2$}.
\ntr{The lengthy proofs of Lemma~\ref{Le:Area_growth_unknown_position} and \ref{Le:Rectangle_recursion_selfsync} are omitted here and can be found in \cite{JS14_Beamforming_LogLog_TR}.}

\tr{
\begin{proof}
The recursions follow from combining Inequality~(\ref{l:eq8}) with~(\ref{l:eq9}).
\begin{eqnarray*}
w_{i+1}
&\stackrel{(\ref{l:eq8})}{=}& \frac{1}{3\sqrt2} w_i h_i
\stackrel{(\ref{l:eq9})}{=} \frac{1}{3\sqrt2} \left(\frac{3}{2\pi^2}\right)^{1/2} \frac{\sqrt{\lambda}}{i+1} w_i^{3/2} 
= \frac{1}{\sqrt{12}\pi} \cdot \frac{\sqrt{\lambda}}{i+1} \cdot w_i^{3/2} \\
h_{i+1} 
&\stackrel{(\ref{l:eq9})}{=}& \sqrt{\frac{3\lambda}{2\pi^2}}\frac{1}{i+1} w_{i+1}^{1/2} 
\stackrel{(\ref{l:eq8})}{=} \sqrt{\frac{3\lambda}{2\pi^2}} \frac{1}{i+1} \left(\frac{w_ih_i}{18^{1/2}} \right)^{1/2} 
\stackrel{(\ref{l:eq9})}{=} 18^{-\frac{1}{4}} \cdot \frac{i+1}{i+2} \cdot h_i^{3/2} 
\end{eqnarray*}
The equations (\ref{l:eq6}-\ref{l:eq9}) can be proven as follows:
\begin{enumerate}
\item[(\ref{l:eq6})]: 
To prove $h_i \leq h_{i+1}$ we insert $h_0$ into Eq.~(\ref{eq:recursion_height_PhaseConversion}) $$h_1 = 18^{-\frac14} \frac{1}{2} \cdot \left(4\sqrt{18}\right)^{3/2} = 4\sqrt{18} = h_0$$
Both factors $h_i^{3/2}$ and $\frac{1+i}{2+1}$ are monotonous increasing, in particular the derivation of the latter is $\left(i+2\right)^{-2}$ which is positive for $i\ge0$. Thus, if $h_1 = h_0$ then it holds that $h_{i+1} \ge h_i$.
\item[(\ref{l:eq7})]: To proof $w_{i+1} \ge w_i$ let us substitute $c_1:=\frac{\sqrt{\lambda}}{\sqrt{12}\pi}$ in Eq.~(\ref{eq:recursion_width_PhaseConversion}).
\begin{eqnarray*}
w_i &=& \frac{c_1}{i+1} \cdot w_{i-1}^{3/2}
\end{eqnarray*}
Here are the first values of $w_i$:
\begin{eqnarray*}
w_0 \\
w_1 &=& \frac{c_1}{2} \cdot w_0^{3/2} \\
w_2 &=& \frac{c_1}{3} \cdot \frac{c_1^{3/2}}{2^{3/2}} \cdot w_0^{\left(3/2\right)^2} \\
w_3 &=& \frac{c_1}{4} \cdot \frac{c_1^{3/2}}{3^{3/2}} \cdot \frac{c_1^{\left(3/2\right)^2}}{2^{\left(3/2\right)^2}} \cdot w_0^{\left(3/2\right)^3} \\
w_4 &=& \frac{c_1}{5} \cdot \frac{c_1^{3/2}}{4^{3/2}} \cdot \frac{c_1^{\left(3/2\right)^2}}{3^{\left(3/2\right)^2}} \cdot \frac{c_1^{\left(3/2\right)^3}}{2^{\left(3/2\right)^3}} \cdot w_0^{\left(3/2\right)^4} \\
\end{eqnarray*}
A closed-form solution for Equation~(\ref{eq:recursion_width_PhaseConversion}) is
\begin{eqnarray}
w_i &=& c_1^{\sum_{k=1}^i\left(3/2\right)^{k-1}} \cdot \prod\limits_{k=1}^i\left(2+i-k\right)^{-\left(3/2\right)^{k-1}} \cdot w_0^{\left(3/2\right)^i} \nonumber\\
w_i &\le& c_1^{2\left(3/2\right)^i-2} \cdot c_2^{-\left(3/2\right)^i} \cdot w_0^{\left(3/2\right)^i} \label{eq:width_sync_closedform}
\end{eqnarray}
for a constant $c_2$ fulfilling the inequation
\begin{eqnarray*}
c_2^{-\left(3/2\right)^i}
&\ge& \prod\limits_{k=1}^i\left(2+i-k\right)^{-\left(3/2\right)^{k-1}} \\
&=& 2^{\sum_{k=1}^i -\left(3/2\right)^{k-1}\cdot\log\left(2+i-k\right)} \\
&=& 2^{\sum_{k=1}^i -\left(3/2\right)^{k-i-1}\cdot\log\left(2+i-k\right) \cdot \left(3/2\right)^i} \ .
\end{eqnarray*}
When substituting $u:=i-k$ we get
\begin{eqnarray*}
c_2^{-\left(3/2\right)^i}
&\ge& 2^{\sum_{u=0}^{i-1} -\left(3/2\right)^{-u-1}\cdot\log\left(2+u\right) \cdot \left(3/2\right)^i}\\
&\ge& \left(2^{c_4}\right)^{-\left(3/2\right)^i}
\end{eqnarray*}
where $c_2 = 2^{c_4}$ can be upper-bounded with
\begin{eqnarray*}
c_4 &=& \sum_{u=0}^{\infty} \frac{\log\left(2+u\right)}{\left(3/2\right)^{u+1}}
\end{eqnarray*}
which converges to $c_4 \approx 3.586$ and we get for the constant
\begin{eqnarray}
c_2 &=& 2^{c_4} = 12.011\ldots
\end{eqnarray}
\begin{figure}
	\begin{center}
	\includegraphics[scale=1.2]{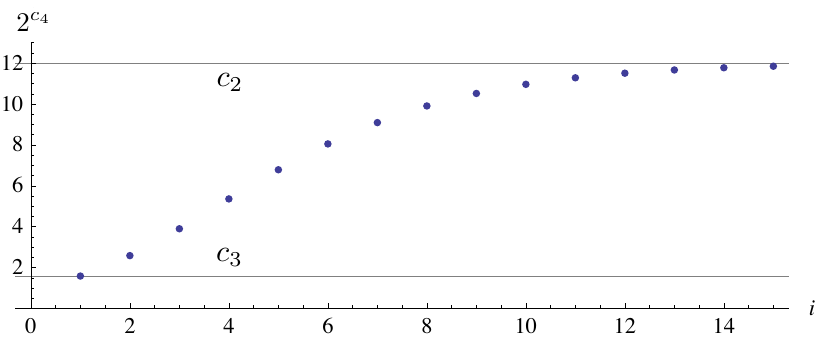}
	\end{center}
	\caption{Constant $c_4$ in approximated closed solution of in Eq.~(\ref{eq:width_sync_closedform}) for width $w_i$.}
	\label{fig:Unicast_Recursion_SelfSync_Width_Constant_Plot}
\end{figure}
Figure~\ref{fig:Unicast_Recursion_SelfSync_Width_Constant_Plot} shows a plot of $2^{c_4}$ for recursion step $i$ with convergence to $c_4$ and a lower bound is also marked for $i=1$ with
\begin{eqnarray}
c_3 &\le& 1.58 \ .
\end{eqnarray}
To satisfy $w_{i+1} \ge w_i$ we have to assure in Eq.~(\ref{eq:width_sync_closedform}) that the initial value $w_0$ compensates from the beginning the limiting factors with $c_1$ and $c_2$. Thus,
\begin{eqnarray}
w_0 &\ge& \frac{c_2}{c_1^2} = \frac{12\pi^2\cdot c_2}{\lambda} \text{ for a constant } c_2 = 12.011
\end{eqnarray}

\item[(\ref{eq:height_sync_closedform})]: Here are the first values from the recursion of the height in Eq.~(\ref{eq:recursion_height_PhaseConversion}) with the constant $c_6 = 18^{-1/4}$
\begin{eqnarray*}
h_0 \\
h_1 &=& \frac{2}{3}c_6 \cdot h_0^{3/2}\\
h_2 &=&\frac{3}{4}c_6 \cdot \left(\frac{2}{3}\right)^{3/2}c_6^{3/2} \cdot h_0^{\left(3/2\right)^2}\\
h_3 &=&\frac{4}{5}c_6 \cdot  \left(\frac{3}{4}\right)^{2/3}c_6^{2/3} \cdot \left(\frac{2}{3}\right)^{\left(3/2\right)^2}c_6^{\left(3/2\right)^2} \cdot h_0^{\left(3/2\right)^3}\\
\end{eqnarray*}
The closed-form solution for $h_i$ is therefore
\begin{eqnarray*}
h_i 
&=& c_6^{\sum_{k=0}^{i-1}\left(2/3\right)^k} \cdot \prod\limits_{k=1}^i \left(\frac{i+1}{i+2}\right)^{i-k} \cdot h_0^{\left(3/2\right)^i} \\
&=& c_6^{2\left(3/2\right)^i-2} \cdot \left(\frac{i+1}{i+2}\right)^{\frac{1}{2}\left(i-1\right)i} \cdot h_0^{\left(3/2\right)^i} \ .
\end{eqnarray*}


%
\item[(\ref{l:eq8})]:  $w_{i+1} \leq \frac{1}{3\sqrt2} w_i h_i$

Now $h_0 = \sqrt{\frac{3\lambda}{2\pi^2} w_0}$, which implies
$ w_0 h_0 = \sqrt{\frac{3\lambda}{2\pi^2}} w_0^{3/2}$.
Inserting $w_{i+1}$ of Equation~(\ref{eq:recursion_width_PhaseConversion}) gives
\begin{eqnarray*}
\frac1{\sqrt{12}\pi}\cdot \frac{\sqrt{\lambda}}{i+1} \cdot w_i^{3/2} &\le& \frac{1}{3\sqrt2} w_i h_i \\
\Rightarrow \frac{1}{\pi^2}\frac{\lambda}{\left(i+1\right)^2} \cdot w_i &\le& h_i^2 \text{ and replacing $h_i^2$ with Eq.~(\ref{l:eq9})} \\
\Rightarrow w_i &\le& \frac{3}{2}  w_{i} \Rightarrow \text{ true.}
\end{eqnarray*}
%
\item[(\ref{eq_ratio2})] We can show the inequation $h_i \le w_i $ by comparing the closed solutions of $w_i$ and $h_i$. Insertion of the lower bound for $w_0$ in Eq.~(\ref{eq:width_sync_closedform}) gives
\begin{eqnarray}
w_i &\ge& \left(\frac{\lambda}{12\pi^2}\right)^{\left(3/2\right)^i-1} \cdot c_2^{-\left(3/2\right)^i} \cdot \left(\frac{8e\cdot 12\pi^2\cdot c_2}{\lambda}\right)^{\left(3/2\right)^i} \nonumber \\
&=& \frac{12\pi^2}{\lambda} {\left(2e\right)}^{\left(3/2\right)^i} \cdot 4^{\left(3/2\right)^i}  \label{eq:closed_solution_width_selfsync}
\end{eqnarray}
Insertion of $h_0$ into Eq.~(\ref{eq:height_sync_closedform}) gives
\begin{eqnarray}
h_i &=& \sqrt{18}^{-\left(3/2\right)^i + 1} \cdot \left(\frac{i+1}{i+2}\right)^{\frac{1}{2}\left(i-1\right)\cdot i} \cdot \left(4\cdot \sqrt{18}\right)^{\left(3/2\right)^i}
= \sqrt{18} \cdot 4^{\left(3/2\right)^i}  \label{eq:closed_solution_height_selfsync}
\end{eqnarray}
For $w_i \ge h_i$ it follows $\lambda \le 8\cdot \pi^2 \cdot c_5^{3/2} \approx 353$ which is true.

\item[(\ref{l:eq9})] For proving $h_{i}^2 \leq \frac{3}{2\pi^2} \frac{1}{(i+1)^2}  \lambda w_{i}$, we insert the the closed solutions of $w_i$ (Eq.~(\ref{eq:closed_solution_width_selfsync})) and $h_i$ (Eq.~(\ref{eq:closed_solution_height_selfsync})).
\begin{eqnarray*}
18 \cdot 4^{2\left(3/2\right)^i} &\le& \frac{3}{2\pi^2} \cdot \frac{1}{\left(i+1\right)^2} \lambda \cdot \frac{12\pi^2}{\lambda} \cdot \left(8c_5\right)^{\left(3/2\right)^i} \\
\Rightarrow \left(i+1\right)^2 &\le& c_5^{\left(3/2\right)^i} \\
\Rightarrow 2 \cdot \log\left(i+1\right) &\le& \left(\frac{3}{2}\right)^i \cdot \log c_5 \\
\Rightarrow \frac{\left(3/2\right)^{i+1}}{\log\left(i+1\right)} &\ge& \frac{4\log c_5}{3} \\
\Rightarrow \left(i+1\right) \log\left(3/2\right) - \log\log\left(i+1\right) &\ge& \log\left(4/3\right) + \log\log\left(c_5\right) \\
\Rightarrow i \cdot \log\left(3/2\right) - \log\log\left(i+1\right) &\ge& \log\left(8/9\right) + \log\log\left(c_5\right)
\end{eqnarray*}
This inequality cannot be solved for $i$ in closed form and therefore we compute a lower bound for the constant $c_5$ by analyzing the following function:
\begin{eqnarray}
\frac{c_5^{\left(3/2\right)^i}}{ \left(i+1\right)^2} &\ge& 1 \label{eq_lowerbound_c5}
\end{eqnarray}
and setting the derivation for variable $i$  equals zero gives
\begin{eqnarray*}
0 &=&\frac{\left(3/2\right)^i \cdot c_5^{\left(3/2\right)^i}\cdot \log\left(3/2\right)\log\left(c_5\right)}{\left(i+1\right)^2} - \frac{2c_5^{\left(3/2\right)^i}}{\left(i+1\right)^3}\\
\frac{2}{\left(i+1\right)} &=& \left(3/2\right)^i \cdot \log\left(3/2\right)\log\left(c_5\right) \\
i_0 &=& \frac{\text{ProductLog}\left(\frac{3}{\log\left(c_5\right)}\right)-\log\left(3/2\right)}{\log\left(3/2\right)} \\
\end{eqnarray*}
Here, the function $\text{ProductLog}\left(x\right) = w$ is the inverse function of $x \mapsto w e^w$. The value  $i_0$ is the location of the minimum and to find the minimum it needs to be substituted into Inequality~(\ref{eq_lowerbound_c5}).  So, this  implies
\begin{eqnarray}
c_5 &\ge& e
\end{eqnarray}
since $\frac{c_5^{\left(3/2\right)^{i_0}}}{ \left(i_0+1\right)^2} = 1.00208 > 1$.
\end{enumerate}
\end{proof}
} 

It remains to show how to inform the first rectangle. \tr{For this, we use the broadcast algorithm of  \cite{JS13_Beamforming_Line}.}

\begin{lemma}
A start phase of $\OO\left(-\log\lambda\right)$ rounds allows to inform an initial area of nodes with 
$w_0 > \frac{72}{\lambda}$, $h_0 \geq \sqrt{18}$, $h_0^2 \leq \frac14 \lambda w_0$, and $h_0 \le w_0$.
\end{lemma}
\ntr{\vspace{-1.5em}}
\begin{figure}[hbt]
	\begin{center}
	\includegraphics[scale=0.6]{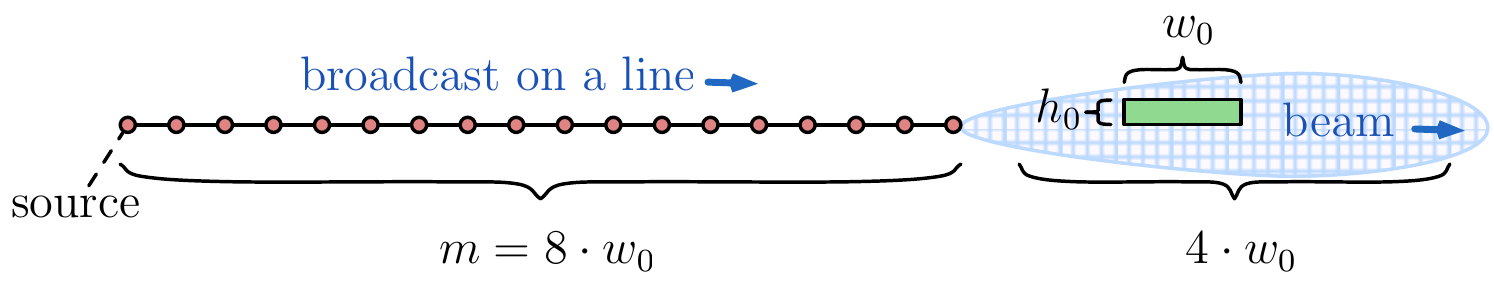}
	\end{center}
	\ntr{\vspace{-1em}}
	\caption{\tr{In an initial phase, a broadcast on a line with $m=8w_0$ nodes is performed followed by a last hop of cooperative beamforming from the line of senders to the first rectangle with dimensions $w_0\times h_0$.}\ntr{Initial phase with a broadcast on the line to $m=8w_0$ nodes followed by a last hop of cooperative beamforming to the first rectangle with dimensions $w_0\times h_0$}}
	\label{fig:Initial_Phase_Broadcast}
\end{figure}

\begin{proof}
%
%
%
To inform the first rectangle with dimensions $w_0 \times h_0$, we first inform $8w_0$ subsequent nodes placed on a line which together can inform and synchronize all nodes  in the first rectangle with cooperative beamforming (compare Fig.~\ref{fig:Initial_Phase_Broadcast}).
To initially inform a line of $m=8w_0$ senders, we use the exponential broadcast algorithm of \cite{JS13_Beamforming_Line}, which informs $m$ nodes placed on a line in $\OO(\log m)$ rounds. Note that the exponential broadcast algorithm has informed at least  $(3/2)^i$ nodes after round $i$.\tr{ We choose $m$ large enough that this line can inform a rectangle of dimensions $w_0 \times h_0$ in distance $w_0$.} We choose $m=8w_0$ which results in a runtime $k\cdot\log\left(\frac{8\cdot72}{\lambda}\right)=\OO\left(-\log\lambda\right)$ rounds for some constant $k$. 
Then, $8w_0$ nodes are in phase to inform not only the next $4w_0$ nodes on the line but also all other nodes in the beam including a rectangle with dimensions $w_0 \times h_0$. However, there will be a phase shift  for the nodes of the rectangle, which are not on the line. By Lemma~1 this offset attenuates the signal by a factor of at most $\frac1{\sqrt2}$. Therefore, all nodes of this initial rectangle receive the message.
\ntr{We can compute the delay error for each node with an invariant of Equation~(\ref{eq:rectangle_delay}) where the sender area is reduced to a line.
}\tr{Analogous to Equation~(\ref{eq:rectangle_delay}), we can compute the delay error for each node in the first rectangle placed at $\vect{r}=\left(r_x,r_y\right)$ with
\begin{eqnarray}
\psi\left(i,\vect{r}\right) &=& \frac{1}{f} + \frac{1}{2\pi f} \arg[\sum_{x=0}^{8w_0-1} \frac{e^{-j2\pi\left(\sqrt{\left(x-r_x\right)^2+r_y^2}-r_x\right)/\lambda}}{\sqrt{\left(x-r_x\right)^2+r_y^2}}  ] \ . \label{eq:line_delay}
\end{eqnarray}
} 
\end{proof}


The above lemmas lead to our main result of  $\OO(\log \log n)$ unicast.

\begin{theorem} \label{t-one} \label{t:one}
Given $n$ nodes in a grid equipped with a trans\-cei\-ver with wavelength $\lambda\le \frac{1}{2}$, placed within unit distance and possessing a transmission power only to reach each neighbor, any node can send a message to any other node in $\OO(\log \log n - \log \lambda)$ rounds.
\end{theorem}
\begin{proof}
The basic idea is, first to route on the $x$-axis until the correct $y$-coordinate has been reached and then to relaunch the algorithm orthogonally on the $y$-axis. Then, the claim follows by the above lemmas. 
%
%
%
%
%
%
\end{proof}

The energy is given by the sum of sending nodes, i.e. $\sum_{i=1}^r w_i h_i$ for $r$ rounds, since each node sends with constant energy.
Now, $w_i h_i = \OO(w_{i+1})$, where $w_{r+1}= \OO(d)$ and $w_i$ grows double exponentially. So, for the sum of transmission energy the last term asymptotically bounds the sum.
\begin{corollary}
The overall transmission energy consumed by the $\OO(\log \log n)$ unicast algorithm for sending a message over distance $d$ is $\OO(d)$.
\end{corollary}

Now, we apply this observation to randomly placed nodes in the grid. First, we establish a bound on the minimum number of nodes in some area.

\begin{lemma}\label{le:random_node_placement}
Given $n$ nodes randomly distributed in a square of area $n$ with transmission range $k \sqrt{\log n}$ for some constant $k $.
In every geometric object inside a square of an area of at least $k^2 \log n$ lie at least $\log n$ nodes with high probability, i.e. $1-n^{-\ell}$ for some constant $\ell$.
\end{lemma}
\tr{
\begin{proof}This follows from a straight-forward application of the Chernoff bound. Let $X$ denote the number of nodes in the square. Then, the probability for a node lying in it is $p=\frac{k ^2 \log n}{n}$. The expected number of nodes is $\mu= pn = k ^2 \log n$.  Now, we use for $0<\delta<1$.
\begin{equation}
 \text{\bf P}(X \leq (1-\delta) \mu)  \ \leq \ e^{-\frac{\delta^2}{2} \mu}
\end{equation}
For $\delta = \frac{\sqrt{\ell \ln 4}}{k } $ we have
\begin{equation}
 \text{\bf P}(X \leq \left(1-\frac{\sqrt{\ell \ln 4}}{k } \right) k ^2 \log n )  \ \leq \ n^{-\ell} \ .
\end{equation}
For $k \geq  \sqrt{4+ \ell \ln4}$ we have $\left(1-\frac{\sqrt{\ell \ln 4}}{k } \right) k ^2 \geq 1$
\begin{equation}
 \text{\bf P}(X \leq \log n )  \ \leq \ n^{-\ell} \ .
\end{equation}
\end{proof}
} 

If the transmission distance is asymptotically smaller, the network is disconnected with probability 1 in the limit \cite{gupta1998critical}. 

\begin{theorem}\label{th:Unicast_random_node_placement}
Given $n$ nodes randomly distributed in a square of area $n$ with transmission range $k \sqrt{\log n}$ for some constant $k>0$. Then, for wavelength $\lambda \geq \frac{3k}{\sqrt{\log n}}$ a node can send a message to any other node in time $\OO(\log \log n)$  with high probability, i.e. $1-n^{-\OO(1)}$.
\ntr{The overall transmission energy for sending a message over distance $d$ is $\OO(d \log n)$.}
\end{theorem}
\tr{
\begin{proof}
We use the above observation which lower-bounds the number of nodes in the transmission range of the start node $s$ as $\log n$. 

Now we consider a $k\sqrt{\log n} \times k\sqrt{\log n}$ square around the start node. We need to do a preparation step where we inform a rectangle satisfying the rectangle properties (\ref{l:eq2})-(\ref{l:eq5}). Consider a rectangle $w_1 \times h_1$ in distance $w_1$ from the start square. 

We choose
\begin{eqnarray}
w_1 & = &  \frac13 k \log^{3/2} n\\
h_1 & =& k \sqrt{\log n}
\end{eqnarray}
and prove that within one hop this rectangle can be informed from the start square which fulfills the rectangle properties  (\ref{l:eq2})-(\ref{l:eq5}) and can serve as a start rectangle for the double exponential growth of Theorem~\ref{t-one}. We assume that all these nodes have position information which they can use to adapt the phase in order for the second phase of the algorithm.

This rectangle is in reach of the start square, since we have at least $\log n$ nodes (with high probability). These nodes have transmission range $k \sqrt{\log n}$ each, since $w_1 \leq \frac13 k \log^{3/2}n$.

Inequality (\ref{l:eq5}) states that $h_1^2 \leq \frac14 \lambda w_1$. Since $\lambda \geq \frac{3k}{\sqrt{\log n}}$ we have
$$ h_1 = k^2 \log^2 n 
\leq \lambda w_1 \ .$$
%
%
The number of nodes in the $w_1 \times h_1$ rectangle has increased to $\Omega(\log^{2} n)$ with high probability. 
From now on, the rest follows by the double exponential growth argument analogously to  Theorem~\ref{t-one}, where each step is successful with high probability. This can be proven by Chernoff bounds, since the transmission distance is a factor $\OO(\sqrt{\log n})$ larger than in the grid model.
\end{proof}
} 
\ntr{The proofs of Lemma~\ref{le:random_node_placement} and Theorem~\ref{th:Unicast_random_node_placement} are presented in \cite{JS14_Beamforming_LogLog_TR}.}

\tr{Now each node sends with energy $\OO(\log n)$, which is proportional to the square of the transmission range. 
Like in the first Corollary the number of sender nodes is again $\OO(d)$. Therefore we have the following energy consumption.
\begin{corollary}
The overall transmission energy in the randomly positioned case for sending a message over distance $d$ is $\OO(d \log n)$.
\end{corollary}}


\ntr{
The transmission time of each hop in a multi-hop algorithm consists of the transmission delay between sender and receiver, the transmission of the message, and processing the message at the receiver. The following theorem shows that the double exponential growth of the transmission distance in the $\mathcal{O}\left(\log\log n\right)$ unicast algorithm is such large that the transmission delay dominates the propagation speed up to a constant factor (The proof is presented in \cite{JS14_Beamforming_LogLog_TR}).
\begin{theorem}
For $\lambda \in \Omega(1)$ and a quadratic grid with $n$ nodes with unit node distance and unit transmission distance, it is possible to send a message from any node to any other node with a speed of $c (1-o(1/n))$, where $c$ is the speed of light.
\end{theorem}
} 

\tr{
\section{Converging towards the speed of light}

For broadcast on the line we have presented a method which needs $\OO(\log n)$ rounds \cite{JS13_Beamforming_Line}. The processing time at each relay node consists of receiving the message, analyzing it, and re-sending it, which we denote by $t_0$. Note that $t_0$ is a constant. Let us denote the node distance from the start node by $d$ and $c$ denotes the speed of light as the signal speed. 
\begin{lemma} \label{le:line_broadcast_speed}
For broadcast on the line, the maximum transmission speed is at most $\frac{1}{\sqrt{2}}c$ which is a constant slower than speed of light $c$. 
\end{lemma}
\begin{proof}
In each round $i$ the transmission distance increases exponentially by $d_i= b^i$ for some basis $b \in (1,2)$. Then in round $r= \lceil \log_b d \rceil$ the target is reached.

So, the overall time $T(d)$ is 
\begin{eqnarray*}
 T(d) &\le& r t_0 + \sum_{i=1}^{r} \frac{d_i}{c} 
=  r t_0 +   \frac1c \frac{b^{r+1}-1}{b-1} \ .
\end{eqnarray*}
Since $db \leq b^{r+1} \leq d b^2$ we have 
$$ T(d) \geq t_0  \lceil \log_b d \rceil + d \frac1c \frac{b-\frac1d}{b-1} \ .$$
Therefore the transmission velocity $v(d) = d/T(d)$ is at most
\begin{eqnarray*}
 v(d) 
&\leq& c \left(1- \frac{1}{b} \pm o(1)\right) \ .
\end{eqnarray*}
So, the maximum speed of transmission on the line is a constant fraction of the speed of light. 
\end{proof}

In two dimensions the situation is different. However, the unicast algorithm presented in Theorem~\ref{t:one} sends a message along the $x$-axis and then along the $y$-axis and this detour reduces the transmission speed to at most $\frac{1}{\sqrt{2}} c$.

%

\begin{theorem}
For $\lambda \in \Omega(1)$ and a quadratic grid with $n$ nodes with unit node distance and unit transmission distance it is possible to send a message from any node $s$ to any other node $w$ with a speed of $c (1-o(1/n))$.
\end{theorem}
\begin{proof}
We use the same construction as in Theorem~\ref{t:one}, but now we tilt the rectangles such that the beamforming is straight from $s$ to $w$. The number of nodes in the rectangles does not change except to some boundary effects, the influence of which is negligible. The starting rectangle needs a width of $w_0 = \Omega(1/\lambda)$. Since, $\lambda \in \OO(1)$ we can inform all nodes of this rectangle in constant time  sequentially by single hop messages and add delay instructions to set up beamforming in the starting rectangle.

Then, the distances $w_i$ grow double exponentially, i.e. $w_i = w_{i-1}^b = (w_0)^{b^i}$ for some $b>1$ and $w_0>1$. The number of rounds is $r = \OO(1) + \log_{d}\log_{w_0} d$ for distance $d$. Note that 
\begin{eqnarray*}
w_i-\sum_{j=0}^{i-1} w_{j} 
&=& (w_0)^{b^i} - \sum_{j=0}^{i-1} (w_0)^{b^j} 
\geq  (w_0)^{b^i} \left(1 -  i (w_0)^{- b^i (1-1/b)} \right) \\
 &=&  (w_0)^{b^i} (1-o(1))
\end{eqnarray*}

Compared to the signal speed $\frac{d}{c}$ we get two kinds of delays: one for the message handling in each round, i.e. $\OO(\log \log d)$. The other one for waiting until a rectangle of size $w_i$ is reached before the last round and this rectangle can relaunch the beamforming. So, in each round we have a message delay of $\frac{w_i}{c}$ for all $i<r$.  The last hop $w_r$ dominates all other rounds, if  we adapt the second last step by using a shorter beamforming step if necessary. This guarantees that the target is reached within the rectangle and that the last inequality above holds for $w_r$.

Note that $d = w_r + \sum_{i=1}^{r-1} 2 w_i + w_0$ and therefore $\sum_{i=0}^{r-1} w_i = o(d)$. So, the overall time for the message transmission is
$$ T(d) = \frac{d + o(d) + \OO(\log \log d)}{c} = \frac1c d (1+o(1))\ .$$
So, the message velocity is
$$ v(d) = \frac{d}{T(d)} = c (1-o(1)) \ .$$
\end{proof}
} 

\tr{
\section{Upper bound for Electromagnetic Field Strength}

The unprecedented long reach of the rectangular field begs the question whether the received signal energy might become too strong to be tolerated. The following lemma shows that the signal strength, which is proportional to the square root of the received power, grows rather moderately.
\begin{lemma}
In a network with $\sqrt{n}\times\sqrt{n}$ nodes, Unicast~\RNum{1} and \RNum{2} produce signal amplitudes $\mathcal{O}\left(\max\left\{ \ln n, \lambda^{1/3}\cdot n^{1/6} \cdot  \ln\frac{n}{\lambda} \right\}\right)$.
\end{lemma}
\begin{proof}
In our setting,  $n$ nodes are placed in a grid in the plane with grid distance $1$ and corresponding dimensions of the network $\sqrt{n}\times\sqrt{n}$. 
Then, the rectangle of the last step can have maximum width $w_{\ell+1}=\sqrt{n}/2$ with distance $\sqrt{n}/2$ to the sender rectangle with dimensions $w_{\ell}\times h_{\ell}$ which we can compute with the equations of Lemma~\ref{le:rectangle_dim}.
\begin{eqnarray*}
w_{\ell+1} &=& \frac{72}{\lambda}\cdot \left(\left(\frac{\lambda}{72} w_0\right)^{\left(3/2\right)^\ell}\right)^{\left(3/2\right)} 
 \ \  \Rightarrow \ \  w_{\ell} = w_{\ell+1}^{2/3} \cdot \left(\frac{72}{\lambda}\right)^{1/3} \\
%
\end{eqnarray*}
Substituting $w_{\ell+1}$ with $\sqrt{n}/2$ gives
\begin{eqnarray*}
w_{\ell} 
&=& n^{1/3} \cdot \frac{18^{1/3}}{\lambda^{1/3}} \ .
\end{eqnarray*}
Using Equation~(\ref{l:eq5}) we get the height of the rectangle
\begin{eqnarray*}
h_\ell 
&=& \frac{1}{2} \lambda^{1/2} \cdot w_\ell^{1/2}
= \frac{1}{2} \lambda^{1/2} \cdot n^{1/6} \cdot \frac{18^{1/6}}{\lambda^{1/6}}
= \frac{18^{1/6}}{2} \lambda^{1/3} \cdot n^{1/6}  \ .
\end{eqnarray*}
We can upperbound the signal amplitude at the end of one horizontal line in the rectangle with $w_\ell$ senders with
\begin{eqnarray*}
\abs{h_{\text{line}}} &\le& 2 \cdot \sum_{i=1}^{w_\ell} \frac{1}{i} \le 2 + 2\cdot \ln\left(w_\ell\right) 
= 2 + \frac{2}{3}\ln\left(\frac{18}{\lambda}\right) + \frac{2}{3} \ln\left(n \right) \ .
\end{eqnarray*}
Now we consider the nearest node to the sender rectangle in the middle of the sender beam. We can upperbound the signal amplitude by adding the signal strength of all $h_{\ell}$ lines with length $w_{\ell}$. With the beamforming setup and $w_\ell \gg h_\ell$ the phase error will be rather small and the bound will be tight. Then we have
\begin{eqnarray*}
\abs{h_{\text{rect}}}
&\le& h_{\ell} \cdot \abs{h_{\text{line}}}
= \frac{18^{1/6}}{2} \lambda^{1/3} \cdot n^{1/6} \cdot \left(2 + \frac{2}{3}\ln\left(\frac{18}{\lambda}\right) + \frac{2}{3} \ln\left(n \right) \right) \\
\abs{h_{\text{rect}}}
&=& \mathcal{O}\left(\lambda^{1/3}\cdot n^{1/6} \cdot  \ln\frac{n}{\lambda} \right) \ .
\end{eqnarray*}
Thus, the maximum signal strength of the unicast algorithm is polynomial. 

For the final result, we have also to consider the case of the initial phase, when the line broadcast has been finished. For $\sqrt{n} \le 12w_0 =\frac{12\cdot 72}{\lambda}$ we are in the initial phase and therefore the amplitude is
\begin{eqnarray*}
 2+2\ln\left(\frac{2}{3}\sqrt{n}\right) = 2+ 2\ln\frac{2}{3} + \ln n = \mathcal{O}\left(\ln n\right) \ .
\end{eqnarray*}
Summarizing, we get an asymptotic upper bound of
\begin{eqnarray*}
h_{\text{Unicast~\RNum{1}}} \in \mathcal{O}\left(\max\left\{ \ln n, \lambda^{1/3}\cdot n^{1/6} \cdot  \ln\frac{n}{\lambda} \right\}\right) \ .
\end{eqnarray*}
%

In Unicast~\RNum{2}, we have chosen the initial rectangle with dimensions $w_0\times h_0$ in such a way, that in Eq.~(\ref{eq:recursion_width_PhaseConversion}), which states the recursion of the rectangle width,  factor $w_i^{3/2}$ compensates factor $\left(i+1\right)^{-1}$ right from the start with width $w_0$. The same holds for the recursion for the height of the rectangle. Thus, although the rectangles  of Unicast~\RNum{2} compared with the rectangles of Unicast~\RNum{1} have a larger width to satisfy the maximum phase error, for the asymptotic signal strength we observe
\begin{eqnarray*}
h_{\text{Unicast~\RNum{1}}} \in \mathcal{O}\left(h_{\text{Unicast~\RNum{2}}}\right) \ .
\end{eqnarray*}

\end{proof}
} 
\section{Lower Bound for Time}
\ntr{
We now investigate the principal lower bound of rounds for disseminating a message in a two-dimension grid when each node has constant power $P$ and an omnidirectional antenna in the line-of-side path-loss model.} \tr{
We will now investigate the principal bounds for time delay of disseminating a message in a two-dimension grid. For this, we concentrate on the question, how many rounds it takes at minimum to reach a node in the Euclidean distance $d$, when in the first round only one node was informed.

The key question for the lower bound for time is, up to when we can safely assure that a node $v$ has not received the message, yet. This is the case when all super-positioned signals cannot be distinguished from the background (or internal) noise. 

The super-positioned signal received at $v$ is $$E_v = \sum_{u \in S}  \frac{s_u}{\dist{u-v}}\ ,$$ where $s_u = a_u  e^{j\phi_u}$ is the signal produced at $u$ with a bounded amplitude $a_u$ and phase shift $\phi_u$. 

The energy $P_v$ of the received signal is proportional to the absolute value of the squared signal size 
$$P_v = |E_v^2| = \left|\textstyle\sum_{u \in S}  \displaystyle\frac{s_u}{\dist{u-v}}\right|^2\ .$$ 
If this term is smaller than a constant $c_n$ we assume no signal can be received.

}The following theorem shows the time optimality of our $\OO(\log \log n)$ unicast algorithm.
\begin{theorem}
In a  grid with $n$ nodes with constant  transmission power, every unicast message takes at least $\Omega(\log \log n)$ rounds to reach its destination.
\end{theorem}
\begin{proof}
Let $u$ be the start node and let $C_d:= \{v \in V : |u,v|\leq d\}$ denote all nodes within Euclidean distance at most $d$ from $u$.

Now in round $i$, let $d_i$ be the distance of the farthest node in this round carrying the (or some parts of the) message. Now consider a node $v$ in distance $d'\gg d_i$.

The received energy is bounded by 
\begin{eqnarray*}
P_v 
= |E_v^2| \leq \left|\sum_{u \in C_{d_i}}  \frac{s_u}{\dist{u-v}}\right|^2
\leq \left( \sum_{u \in C_{d_i}} \frac{\sqrt{P}}{d'-d_i} \right)^2
\leq P \frac{|C_{d_i}|^2}{(d'-d_i)^2}  \tr{\ ,}\ntr{\ .}
\end{eqnarray*}
\tr{where $P$ is the maximum transmission power of each node (a constant). }In order to receive the signal, this power must be larger than a constant $\tau >0$. We want to investigate the case when we cannot receive a signal, i.e. \ntr{$P_v \le \tau$.}
\tr{$$  P \frac{|C_{d_i}|^2}{(d'-d_i)^2} \leq \tau \ .$$}
Then, $ d' \geq d_i + |C_{d_i}| \sqrt{\frac{\tau}{P}}$
which implies with $|C_d| \leq 2 \pi d^2$ that
\ntr{$$d' \geq d_i + 2 \pi d_i ^2  \sqrt{\tau/P} \ .$$}
\tr{$$d' \geq d_i + 2 \pi d_i ^2  \sqrt{\frac{\tau}{P}} \ .$$} 
From this it follows that 
$d_{i+1} \leq k\cdot  d_i^2$ for a constant $k>0$ and thus
%
%
$$d_{i+1} \leq k^{2^i-1} (d_1)^{2^i}\ .$$ Therefore, it takes at least some $k' \log \log d$ rounds (for a constant $k'>0$) to inform a node in distance $d$.
\end{proof}

\ntr{\vspace{-2em}}
\begin{figure}[htb]
  \centering
  \subfigure[SNR with color range $[\text{orange,white})$ over threshold $\tau$ and $[\text{purple, cyan})$ under $\tau$]  {
    \label{fig:Simulation_UnicastStep_SNR}
    \includegraphics[scale=0.18]{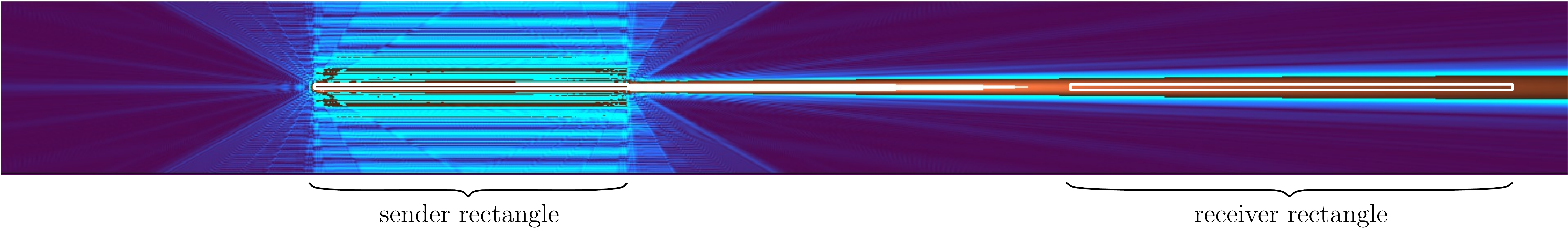}
  }\hspace{0.5em}
  \subfigure[Phase error with angle range $[0,\pi)$ and colors $[\text{black, blue})$]{
    \label{fig:Simulation_UnicastStep_PhaseError}
    \includegraphics[scale=0.18]{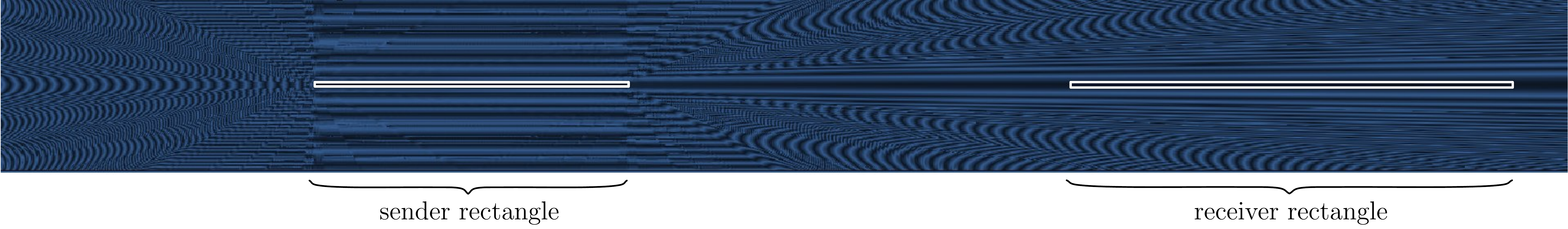}
  }
  \caption{Simulation of beamforming senders which are placed in a rectangle and produce a beam to the right. An animation with varying wavelength $\lambda$ is available at \texttt{www.youtube.com/watch?v=3TJ2Gz8uhbc}.}
  \label{fig:Simulation_UnicastStep}
\end{figure}
\ntr{\vspace{-3em}}

\section{Simulation}


We have simulated cooperative sender beamforming for nodes placed in a rectangle in the plane. The dimensions of the rectangles correspond to Unicast~\RNum{1} (compare Fig.~\ref{fig:Rectangle_Beamforming_Gain}).
Figure~\ref{fig:Simulation_UnicastStep} shows the signal strength respectively phase shift of a 1705$\times$186 grid network with grid distance $1$ (one pixel=1 node) and  the wavelength is $\lambda=0.1$. We see sender beamforming from a  rectangle with 341$\times$6=2046 nodes to a receiver area with 482$\times$7=3374 nodes (the areas are white bordered). 

The first picture \ref{fig:Simulation_UnicastStep_SNR} shows the signal strength where the blue color range depicts amplitudes under the SNR threshold $\tau=1$ and the orange-white color range represents signal strengths over $\tau$. We can spot a sharp beam around the receiver rectangle with a signal over the SNR threshold.\tr{ The light blue lines over and under the sender rectangle indicate  strong non readable interferences for nodes not involved in the Unicast operation. We can also see two side lobes with 45 degree alongside the main beam.}
The second figure \ref{fig:Simulation_UnicastStep_PhaseError} shows the phase shift for synchronized beamforming. The black corridor from sender to receiver rectangle makes clear, that all nodes receiving the message within this corridor will be synchronized for beamforming to the right. The blue lines around the corridor mark a phase shift of $\pi$ and the subsequent next black rays around have a phase error of $2\pi$, i.e. one period $1/f_c$ of carrier frequency $f_c$. Notably, the spatial variation of the phases of the super-posed signal is much smaller than the wavelength (=0.1 pixels)\tr{ and scales with the size of the sender rectangle}.

Figure~\ref{fig:Simulation_Multicast_VarWavelength} shows the beamforming gain for \tr{the} different wavelengths $\lambda\tr{\in\left\{\frac{1}{8},\frac{1}{4},\frac{1}{2},1,2\right\}}$. The $n=2048$ cooperating senders are selected according to Unicast~\RNum{1} and highlighted with an orange rectangle on the left and the signal is over the SNR threshold in the blue colored area. We did not intend to show the special case where the wavelength is an integer multiple of the grid distance and thus added a small $\epsilon$ to the wavelength.
\tr{\begin{figure}[htb]
  \centering
  \subfigure[$n=1000$, height $h^2 = \lambda w$.]{
    \label{fig:Simulation_Multicast_VarWavelength}
    \includegraphics[scale=0.51]{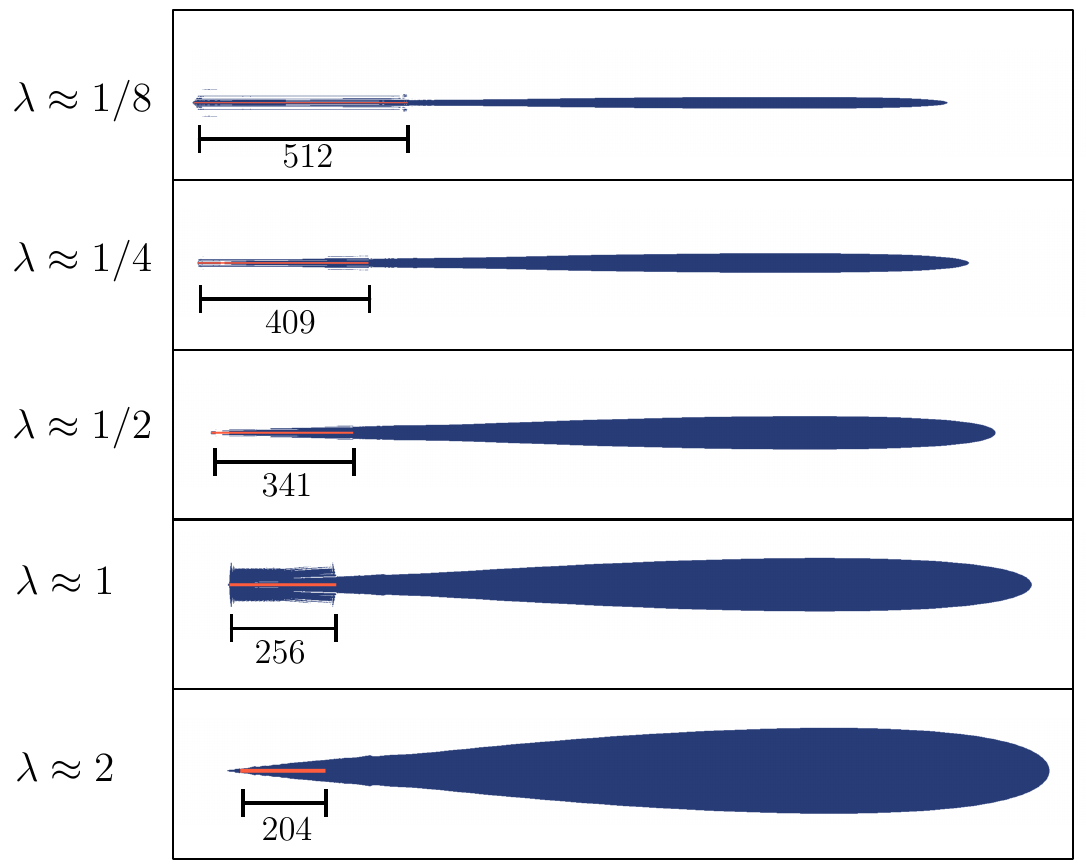}
  }\hspace{0.5em}
  \subfigure[$n=2048$, $\lambda = \frac12$ and varying rectangle sizes.]{
    \label{fig:Simulation_Multicast_VarRatio}
    \includegraphics[scale=0.51]{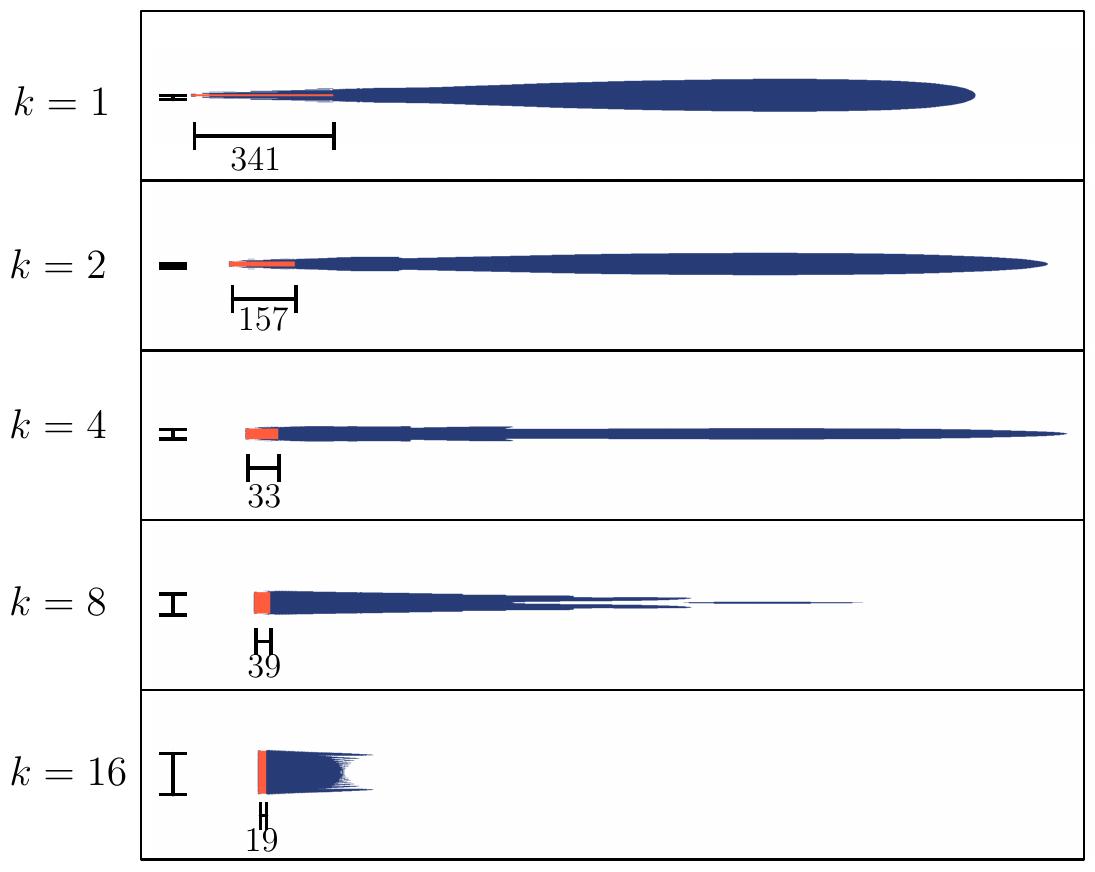}
  } \ntr{\vspace{-0.5em}}
  \caption{Simulation of  $n$ beamforming senders placed in a rectangle (orange colored at the left) which produce a beam to the right.}
  \label{ffig:Simulation}
\end{figure}}
%
The reception distance of the beam is nearly equal to $n$ showing full beamforming gain in the middle of the beam. The height of the beam increases with the wavelength $\lambda$.

In a second experiment, we manipulate for a constant wavelength $\lambda=0.5$ the ratio of the rectangle with factor $k$, i.e. $w:=A/k$ and $h:=A\cdot k$. 
When we increase the height, we can spot two effects. First, the beam is sharper and we cannot reach a rectangle with larger height in the multicast. In the examples $k\ge 4$ the height even shrinks. Second, the perception range decreases and we can only multicast to a short distance.

%

\tr{
We simulated the phase error which occurs in the initial phase (see Eq.~(\ref{eq:line_delay})) when informing the first rectangle with $w_0\cdot h_0$ receivers from a line of $8w_0$ senders in distance $w_0$ with parameters set to  $\lambda=0.1$, $w_0=2\cdot\frac{72}{\lambda}=1440$, and $h_0=\sqrt{\lambda\cdot w_0/4}=6$. 
\begin{figure}
	\begin{center}
	\includegraphics[scale=0.9]{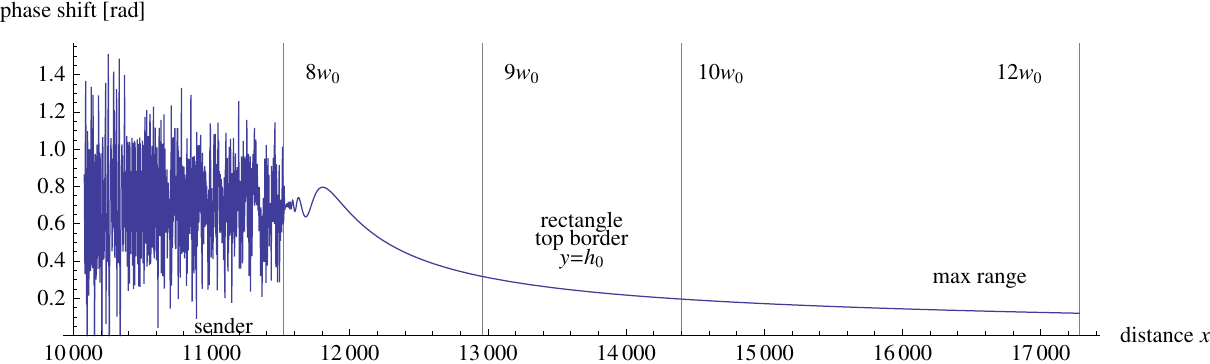}
	\end{center}
	\caption{Phase shift error for broadcasting from a line of $8w_0$ senders to a rectangle of $w_0\cdot h_0$ receivers in distance $w_0$ with parameters $\lambda=0.1$ and $w_0=1440$.}
	\label{fig:InitialPhase_PhaseError}
\end{figure}
Figure~\ref{fig:InitialPhase_PhaseError} shows the phase error compared to the synchronized phase for the coordinates $\left(x,y\right)$ where $y=h_0$ is the top border of the first rectangle. We see that the phase shift around the line of senders for $0\le x\le8w_0$ is arbitrary in the range $\left[0,2\pi\right)$ and for $x\ge9w_0$ the phase shift is smaller than $1/\sqrt{2}$ as assumed.
}

\tr{
Figure~\ref{fig:Unicast_Progress_Graph_Rounds} shows an example for the propagation velocity during the execution of algorithm Unicast~\RNum{1}.
\begin{figure}[hbt]
	\begin{center}
	\includegraphics[scale=0.9]{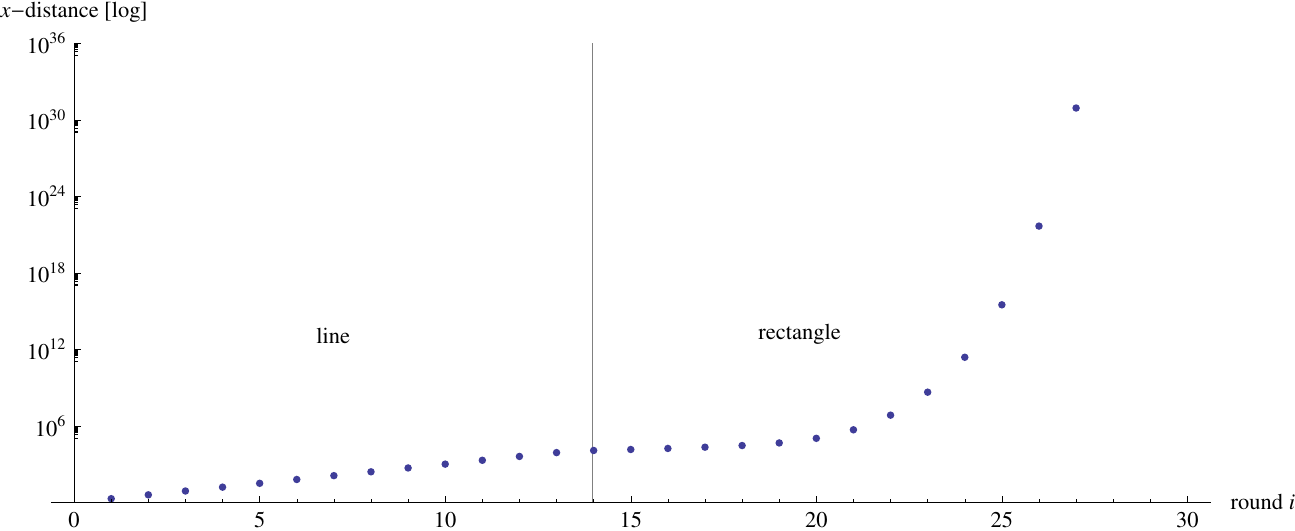}
	\end{center}
	\caption{Progress of the algorithm Unicast~\RNum{1} for $\lambda=0.1$, $w_0=1000$, and the source at $x=0$. The graph shows for round $i$ the $x$-coordinate of the farthest informed node. }
	\label{fig:Unicast_Progress_Graph_Rounds}
\end{figure}
The vertical line separates the initial phase using the line-broadcast with exponential growth from the second phase applying Unicast~\RNum{1} with double exponential growth. The constant slope in linear-log scale suggests an exponential growth in the first phase. When transitioning into the second phase, the slope of the progress first decreases and it takes around 5 rounds that Unicast algorithm can pick up speed and disseminates faster than the exponential growth in the initial phase. But from round $i=26$ on, the information dissemination literally explodes.  But of course, the time for each round increases with hop distance and though speed of light $c$ is the limiting factor as the graph in Figure~\ref{fig:Unicast_Progress_Graph_Time} shows, where the propagation distance $x$ is plotted for time $t$. 
\begin{figure}[hbt]
	\begin{center}
	\includegraphics[scale=0.9]{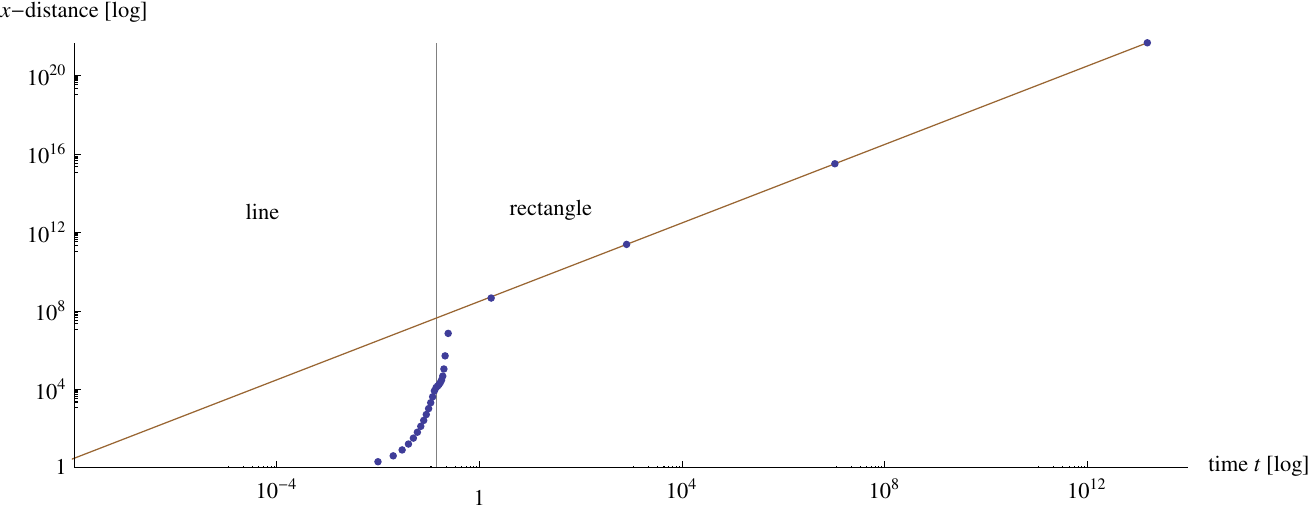}
	\end{center}
	\caption{Progress of the algorithm Unicast~\RNum{1} for $\lambda=0.1$, $w_0=1000$, and the source at $x=0$. The graph shows for round $i$ the $x$-coordinate of the farthest informed node. }
	\label{fig:Unicast_Progress_Graph_Time}
\end{figure}
In this experiment, we assume a distance between nodes of $1$ meter and a processing time of $10^{-2}$ seconds at each relay node. The brown line shows the propagation with speed of light, i.e. one hop broadcast.
} 

\section{Conclusions}

We present a unicast algorithm for ad-hoc networks on a grid with $n$ nodes, which needs only $\OO\left(\log\log n\right)$ rounds for wavelength $\lambda\in \Omega(1)$. This algorithm combines beamforming with multi-hop routing. Beamforming increases the hop distances to a double exponentially growth, i.e. $\OO\left(\left.w_0\right.^{\left(b^i\right)}\right)$ for round $i$. This growing beamforming gain is realized by a set of increasing rectangular areas containing relay nodes. Similar results can be shown for randomly placed nodes in a square, if the transmission range is increased by a factor of $\Omega(\sqrt{\log n})$. The overall transmission velocity of such unicast algorithms converges towards the speed of light and for the grid
we show the optimality of the routing time $\OO\left(\log\log n\right)$.
Such a unicast algorithm does not  asymptotically use more energy than the basic multi-hop algorithm.

Unlike in the one-dimensional case, the wavelength plays a large role in the construction and performance of the algorithm. Short wavelengths \ntr{compared to the node distance }increase the run-time, since it takes longer until the double exponential growth phase begins. For random placement it is not clear how beamforming can be utilized for wavelengths shorter than $\OO(1/\log n)$, while for larger wavelengths our algorithm provides a solution. In the grid, the unicast algorithm has only logarithmic run-time if the wavelength is $\OO(1/n^c)$.



Since we only use beam-formed sending with Multiple Input Single Output (MISO), the main component of the algorithm is to obey a fixed time delay bet\-ween receiving the\tr{ full} message and residing it. Besides this, only a check is needed, whether the relay node is in one of the rectangles necessary for transport. This can be computed from the message header and the position information of the relay node. An exact position information is therefore not necessary. This is an extreme simplification compared to the way beamforming is usually achieved.

\ntr{\vspace{-1em}}

\tr{Note that the wavelength is taken relative to the node density. So, for fixed wavelength the node density plays the same role, where small node distances allow faster unicast.}
\ntr{\begin{figure}[htb]
  \centering
  \subfigure[$n=1000$, height $h^2 = \lambda w$.]{
    \label{fig:Simulation_Multicast_VarWavelength}
    \includegraphics[scale=0.51]{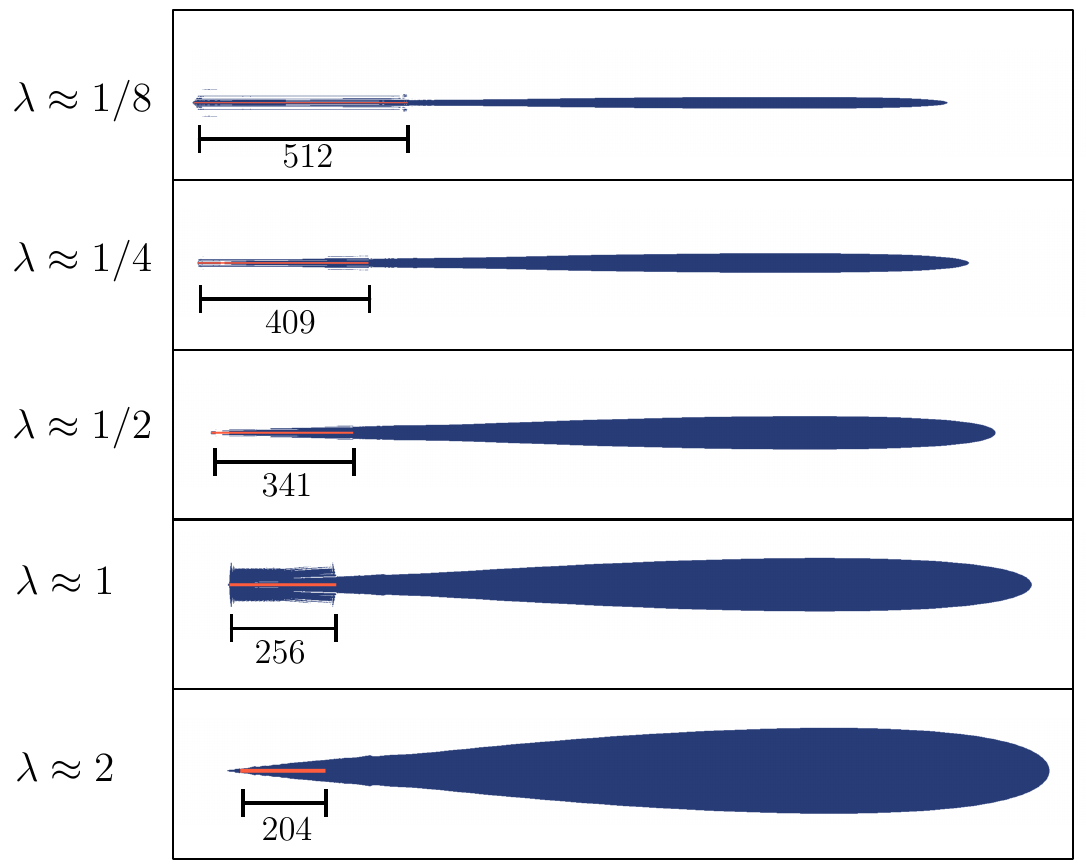}
  }\hspace{0.5em}
  \subfigure[$n=2048$, $\lambda = \frac12$ and varying rectangle sizes.]{
    \label{fig:Simulation_Multicast_VarRatio}
    \includegraphics[scale=0.51]{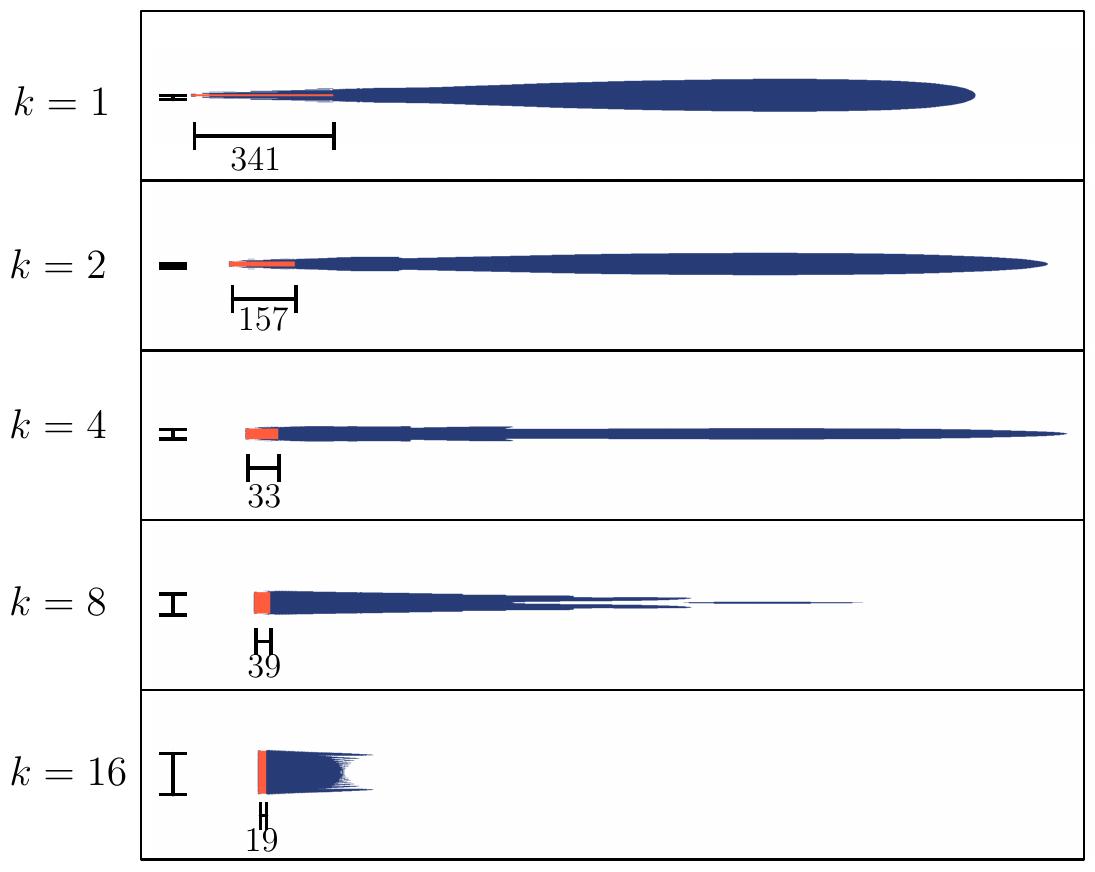}
  } \ntr{\vspace{-0.5em}}
  \caption{Simulation of  $n$ beamforming senders placed in a rectangle (orange colored at the left) which produce a beam to the right.}
  \label{ffig:Simulation}
\end{figure}}
\ntr{\vspace{-2em}}

\bibliographystyle{abbrv}
\bibliography{mimo}

\tr{\section*{Appendix}

\begin{lemma}\label{Marx-a}
For all $x\geq0$\ntr{, it holds that $\frac{x^2}{2} \geq \sqrt{1+x^2}-1$.}
\tr{$$\frac{x^2}{2} \geq \sqrt{1+x^2}-1\ .$$}
\end{lemma}
\begin{proof}
The claim is equivalent to
$\frac{x^2}{2}+1 \geq \sqrt{1+x^2}$.
Squaring both sides yields
\tr{$}$\frac{x^4}{4}+ {x^2} +1 \geq 1+x^2$\tr{$}
which always holds.
\end{proof}

\ntr{
\begin{proofof}{Lemma~\ref{le:rectangle_dim}}
\begin{enumerate}
\item[(\ref{l:eq2})]: Merging Inequality~\ref{l:eq4} with~\ref{l:eq5} gives $h_{i+1} \le \frac{\lambda}{4}\cdot\frac{1}{\sqrt{18}}\cdot w_i h_i \le \frac{1}{\sqrt{18}} h_i^{3/2}$.
Then, $h_i \leq h_{i+1}$ follows from $h_0   \geq \sqrt{18}$.
\item[(\ref{l:eq3})]: $w_i \leq w_{i+1}$
is true if 
$w_0 \geq \frac{72}{w_0}$.
\item[(\ref{l:eq4})]:  $w_{i+1} \leq \frac1{3\sqrt2} w_i h_i $

Now $h_0 = 2 \sqrt{\lambda w_0}$, which implies
$ w_0 h_0 = 2 \sqrt{\lambda} w_0^{3/2}$.
Therefore
\begin{eqnarray*}
\frac1{\sqrt{18}} w_i h_i
& = & \frac{72}{\lambda} \left(\frac{\lambda}{72}
\sqrt{\frac{\lambda }{72}} w_0^{3/2}\right)^{(3/2)^{i}}\\
& = & \frac{72}{\lambda} \left(\frac{\lambda  }{72}w_0\right)^{(3/2)^{i+1}}\\
& = & w_{i+1}
\end{eqnarray*}
\item[(\ref{l:eq4})] $ h_{i}^2 \leq \frac1{4} \lambda w_{i}$:
The following equations finalize the proof.
\begin{eqnarray*}
h_0^2 &=& \frac14 \lambda w_0 \\
\frac1{18} h_0^2 &=& \frac{\lambda w_0}{72} \\
  18 \left(\frac1{18} h_0^2\right)^{(3/2)^{i}}  &=& \frac14 \lambda\left(\frac{72}{\lambda}\right)
 \left(\frac{{\lambda}  }{{72}}w_0 \right)^{(3/2)^{i}} \\
\left.h_i\right.^2 & = & \frac14 \lambda w_i
\end{eqnarray*}
\end{enumerate}
$ $
\end{proofof}
} 

\ntr{
\begin{proofof}{Lemma~\ref{le:random_node_placement}}
This follows from a straight-forward application of the Chernoff bound. Let $X$ denote the number of nodes in the square. Then, the probability for a node lying in it is $p=\frac{k ^2 \log n}{n}$. The expected number of nodes is $\mu= pn = k ^2 \log n$.  Now, we use for $0<\delta<1$.
\begin{equation}
 \text{\bf P}(X \leq (1-\delta) \mu)  \ \leq \ e^{-\frac{\delta^2}{2} \mu}
\end{equation}
For $\delta = \frac{\sqrt{\ell \ln 4}}{k } $ we have
\begin{equation}
 \text{\bf P}(X \leq \left(1-\frac{\sqrt{\ell \ln 4}}{k } \right) k ^2 \log n )  \ \leq \ n^{-\ell} \ .
\end{equation}
For $k \geq  \sqrt{4+ \ell \ln4}$ we have $\left(1-\frac{\sqrt{\ell \ln 4}}{k } \right) k ^2 \geq 1$
\begin{equation}
 \text{\bf P}(X \leq \log n )  \ \leq \ n^{-\ell}
\end{equation}
\end{proofof}
} 

\ntr{
\begin{proofof}{Theorem~\ref{th:Unicast_random_node_placement}}
We use the above observation of Lemma~\ref{le:random_node_placement} which lower-bounds the number of nodes in the transmission range of the start node $s$ as $\log n$. 

Now we consider a $k\sqrt{\log n} \times k\sqrt{\log n}$ square around the start node. We need to do a preparation step where we inform a rectangle satisfying the rectangle properties (\ref{l:eq2})-(\ref{l:eq5}). Consider a rectangle $w_1 \times h_1$ in distance $w_1$ from the start square. 

We choose
\begin{eqnarray}
w_1 & = &  \frac13 k \log^{3/2} n\\
h_1 & =& k \sqrt{\log n}
\end{eqnarray}
and prove that within one hop this rectangle can be informed from the start square which fulfills the rectangle properties  (\ref{l:eq2})-(\ref{l:eq5}) and can serve as a start rectangle for the double exponential growth of Theorem~\ref{t-one}. We assume that all these nodes have position information which they can use to adapt the phase in order for the second phase of the algorithm.

This rectangle is in reach of the start square, since we have at least $\log n$ nodes (with high probability). These nodes have transmission range $k \sqrt{\log n}$ each, since $w_1 \leq \frac13 k \log^{3/2} n$.

Inequality (\ref{l:eq5}) states that $h_1^2 \leq \frac14 \lambda w_1$. Since $\lambda \geq \frac{3k}{\sqrt{\log n}}$ we have
$$ h_1 = k^2 \log^2 n 
\leq \lambda w_1 \ .$$
%
%
The number of nodes in the $w_1 \times h_1$ rectangle has increased to $\Omega(\log^{2} n)$ with high probability. 
From now on, the rest follows by the double exponential growth argument analogously to  Theorem~\ref{t-one}, where each step is successful with high probability. This can be proven by Chernoff bounds, since the transmission distance is a factor $\OO(\sqrt{\log n})$ larger than in the grid model.
\end{proofof}
} 
}

\end{document}